\documentclass[12pt,reqno]{amsart}
\usepackage[left=1.25in,right=1.25in,top=1.25in,bottom=1.25in]{geometry}
\usepackage[colorlinks=true,allcolors=Blue,linkcolor=Blue]{hyperref}
\usepackage{color}
\usepackage[dvipsnames]{xcolor}
\usepackage{amsmath, amsthm, amssymb, amsfonts}
\usepackage{mathtools}
\usepackage{parskip}
\usepackage{natbib}
\usepackage{url}
\usepackage{comment}
\usepackage{enumitem}
\usepackage{graphicx}
\usepackage{tikz}
\usepackage{caption}
\usepackage{subcaption}
\usepackage{pgfplots}
\usepackage{float}
\usepackage[font=small]{caption}
\usepgfplotslibrary{fillbetween}
\usetikzlibrary{patterns}
\pgfplotsset{compat=newest}
\usetikzlibrary{decorations.pathreplacing,angles,quotes}
\tikzset{
	declare function={
		normcdf(\x,\m,\s)=1/(1 + exp(-0.07056*((\x-\m)/\s)^3 - 1.5976*(\x-\m)/\s));
	}
}

\newtheorem{theorem}{Theorem}
\newtheorem{corollary}{Corollary}

\newtheorem{lemma}{Lemma}

\newtheorem{prop}{Proposition}

\theoremstyle{definition}

\begin{document}
\setlength{\footnotesep}{0.75\baselineskip}
\setlength{\footskip}{1.5\normalbaselineskip}

\begin{titlepage}
	\title{Discovery through Trial Balloons}
	\author{\larger \textsc{Eitan Sapiro-Gheiler \\MIT\\ \\ \today}}
	\thanks{Email: eitans@mit.edu
		\\ \indent We thank Robert Gibbons, Stephen Morris, Alexander Wolitzky, and participants in the MIT Theory Lunch for helpful discussions and comments. Thanks especially to Roi Orzach for helping with an earlier version of this project. We acknowledge support from the National Science Foundation Graduate Research Fellowship under Grant No. 1745302.}
	
	\begin{abstract}
		
		A principal and an agent face symmetric uncertainty about the value of two correlated projects for the agent. The principal chooses which project values to publicly discover and makes a proposal to the agent, who accepts if and only if the expected sum of values is positive. We characterize optimal discovery for various principal preferences: maximizing the probability of the grand bundle, of having at least one project approved, and of a weighted combination of projects. Our results highlight the usefulness of trial balloons: projects which are ex-ante disfavored but have higher variance than a more favored alternative. Discovering disfavored projects may be optimal even when their variance is lower than that of the alternative, so long as their disfavorability is neither too large nor too small. These conclusions rationalize the inclusion of controversial policies in omnibus bills and the presence of moonshot projects in organizations.\newline
		
		\noindent\emph{JEL\ Classification:} D21, D83\newline
		
	\end{abstract}
	
	\clearpage
	
	\begingroup
	\let\MakeUppercase\relax
	\maketitle
	\thispagestyle{empty}
	\endgroup
\end{titlepage}

\newpage 
\section{Introduction}
\label{sec:intro}
Many large proposals are composed of multiple smaller projects. For example, the 2022 omnibus appropriations bill passed in the US Congress includes funding for various federal agencies, transportation-related earmarks, and emergency aid for Ukraine;\footnote{As described by the National Conference of State Legislatures \href{https://www.ncsl.org/ncsl-in-dc/publications-and-resources/2022-omnibus-appropriations-bill-a-summary-of-provisions-by-federal-agency.aspx.}{here}.} a grant proposal may contain funding for multiple research designs carried out by a particular lab; or a worker may request their manager's approval for a slate of several projects rather than presenting them one-by-one. In these examples, and many others, even if the value of each \textit{project}, or component issue in the overall \textit{proposal}, is uncertain, there may be a known correlation between projects. After all, an appropriations bill reflects the priorities of the governing party; a scientific lab usually has a well-defined area of expertise; and a worker has a particular role or set of responsibilities. Given the complexity of the full proposal, it both unsurprising and empirically common for proposers in these settings to concentrate attention on one ``headline" project before finalizing the proposal. For instance, the Biden administration campaigned in large part on infrastructure investment, then included infrastructure-related items like clean energy spending alongside changes to Medicare and the tax code in the Inflation Reduction Act.\footnote{The bill's contents are summarized in a White House press release \href{https://www.whitehouse.gov/briefing-room/statements-releases/2022/08/15/by-the-numbers-the-inflation-reduction-act/}{here}.} The choice of which project to focus on thus becomes an important tool for designing such proposals.

We model this setting by considering a principal and an agent who face uncertainty about the agent's value for two correlated projects. The principal chooses which project value, if any, to publicly discover, and then proposes a set of projects to the agent. The proposal is enacted if and only if the agent believes the expected sum of project values is positive. We characterize optimal discovery for several natural forms of the principal's preferences\textemdash maximizing the probability of the grand bundle, of having at least one project approved, and of a weighted combination of projects. Our results highlight the usefulness of \textit{trial balloons}: projects which are ex-ante disfavored but have higher variance than the alternative. In contrast to earlier literature, which generally finds benefits from \textit{building consensus} through information about ex-ante favored projects, in our setting trial balloons have important benefits driven by potential ex-post disagreement on the approver's favored project. Even a disfavored project with lower variance than the more-favored alternative may be preferred if it is moderately disfavored or the prior is highly precise.

The trade-off between trial balloons and building consensus is driven by two forces: \textit{spoiling} and \textit{improvement}. Spoiling is the chance that, by obtaining an unfavorable draw, the principal will lose out on the sure-thing payoff that could have been obtained without generating any information. Improvement is the converse\textemdash a favorable draw may raise the agent's belief sufficiently to get ex-ante disfavored projects approved. Trial balloons have a weaker spoiling effect than other discovery rules, though it may be that spoiling is still likely enough for the principal to prefer no discovery at all. Full discovery\textemdash if permitted\textemdash generally has the strongest improvement effect because it introduces the most variance and thus the highest chance of a sufficiently positive realization. Our main results compare the relative magnitudes of these effects to characterize when trial balloons are optimal. We also provide a partial characterization of optimal discovery when the disfavored project has lower variance, as well as a partial comparison of single-project discovery to full discovery.

Using these characterizations, we describe rich comparative statics of the principal's optimal discovery rule. Because higher variance results in a stronger improvement effect, discovering a trial balloon is always better than discovering a consensus project. However, the effect of a more precise prior on optimal discovery depends on how much the principal values having the trial balloon approved. A trial balloon which is less important than the consensus project benefits from a more precise prior, in the sense that the principal prefers it to no discovery even at higher magnitudes of disfavorability. The opposite is true for a more important trial balloon: the range of magnitudes where the principal discovers it shrinks. Both results are driven by the rates at which spoiling and improvement effects vanish as the prior's precision increases. When the disfavored project has lower variance\textemdash and is therefore no longer a trial balloon\textemdash these results are altered by the consensus project's larger improvement and spoiling effects. In the limit, the disfavored project is discovered only when the magnitude of its disfavor is intermediate. If it is too disfavored, the principal prefers no discovery; if it is only mildly disfavored, the principal prefers discovering the consensus project. Optimal discovery away from the high-precision limit is qualitatively opposite to the trial balloon case. When the principal places more weight on the low-variance disfavored project, the intermediate favorability range where it is discovered grows as prior precision increases; less weight causes it to shrink with higher precision. These sharp differences based on the observable features of the projects available to the principal provide clear testable predictions. 

Finally, we consider extensions of our stylized model to better reflect real-world proposals. Although we are motivated by the idea of a constrained principal who cannot fully design information about the projects, the restriction to discovery of a single project is stark: we relax this choice by allowing either one-shot discovery of both projects or a sequential discovery choice. We also consider noisy discovery, where rather than generating conclusive evidence about the project's value to the agent, the principal may only be able to generate a noisy signal. Lastly, we discuss which results remain tractable in a setting with $N$ projects rather than only two. Our qualitative results about the settings in which trial balloons dominate due to the presence of spoiling effects extend, and emphasize how the desire to keep a consensus option on the table may lead to a focus on disfavored projects.

\section{Related Literature}
\label{sec:relLit}
Our work relates to a broad literature on soliciting and using advice from potentially biased experts. \citet{CT2007} and \citet{D2011} describe hard evidence setting where a principal's disclosure rule is designed to ``build consensus" by suggesting to the agent that some underlying state of the world is favorable. Our model focuses instead on the interplay between different dimensions of a state. Perhaps closest to our work is \citet{BC2007}, where a leader can decide whether or not to publicly gather information; this setup resembles a sequential discovery version of our model, but where the ``dimensions" are signals of the same payoff-relevant state. The richer interaction between dimensions in our model means that, unlike in that work, one-shot discovery is a nontrivial problem in our setting.

A literature in political economy describes ``policy bundling," or combining different projects into a single proposal that receives an up-or-down vote. Many of these models are dynamic and study how today's project choice influences tomorrow's bargaining procedure or options.\footnote{For example, \citet{F1990, CE2017b, L2022}, and \citet{BI2022}.} Our model is instead a one-shot game, though we briefly discuss sequential discovery in Section \ref{subsec:seq}. Another key difference is that rather than using a consensus policy featuring full preference alignment and an ideological dimension featuring disagreement, as is the case in \citet{JM2002} and \citet{CE2013}, we allow for flexible levels of alignment. Indeed, our main results (Theorems \ref{thm:trialBalloon} and \ref{thm:optSingle}) highlight how the ratio of agreement to disagreement shapes the principal's optimal discovery rule. Closer to our work is \citet{CE2017}; however, that model focuses on how to bundle policies (using an information design setup) rather than how to communicate about an exogenous policy bundle.

Although we frame our model as one of symmetric uncertainty, where the principal must publicly \textit{discover} information rather than committing to a \textit{disclosure} rule for private information, this distinction is one of interpretation only. The principal's problem can also be viewed as a constrained version of the standard Bayesian persuasion problem (\citealt{KG2011}, with extensions surveyed by \citealt{K2019}) in which the principal chooses a signal structure to maximize a function of the agent's posterior beliefs. Our first constraint only allows the principal to disclose a single dimension of the multidimensional state, capturing feasibility constraints on the principal (such as the ability to generate only one prototype) as well as attention constraints on the agent (such as the ability to focus on only one spending plan in a large budget bill). Of note is a recent result by \citet{MS2022}, which shows that with a high-dimensional state the optimal unconstrained disclosure rule involves dimension reduction. We view this result as complementary, emphasizing the value of focusing on dimensionally-constrained information design. Our second restriction is that, for the chosen dimension, the principal must choose from an exogenously given set of signal structures rather than having access to any Blackwell experiment. Prior work on constrained information design has considered differently-structured limitations;\footnote{\citet{TT2019} studies a principal facing many copies of the same persuasion problem, whereas \citet{TT2021} considers a single problem with noisy communication.} the closest to our restrictions is the general framework in \citet{B2021}. However, that work focuses primarily on high-level characterizations and computational results rather than a tight description of optimal discovery rules.

\section{Model}
\label{sec:model}
There are two projects $i \in \left\{1, 2\right\}$, which are relevant to both a principal (she) and an agent (he). Let $v = (v_1, v_2) \in \mathbb{R}^2$ represent the agent's value for each project, and $w = (w_1, w_2) \in [0, 1]^2$ with $w_1 + w_2 = 1$ be the relative weights the principal places on them. The principal and the agent share a common prior $F$ from which $v$ is drawn,
\begin{equation*}
	F \sim N\bigg(
	\begin{bmatrix}
		\mu_1 \\ \mu_2
	\end{bmatrix},
	\begin{bmatrix}
		\sigma_1^2 & \rho \sigma_1 \sigma_2\\
		\rho \sigma_1 \sigma_2 & \sigma_2^2
	\end{bmatrix}
	\bigg),
\end{equation*}
so that the marginals $F_1$ and $F_2$ are normally distributed with variances $\sigma_1^2$ and $\sigma_2^2$ respectively, and have correlation $\rho \in (0, 1)$.\footnote{This choice is important for our results; it reflects a world in which projects are pre-screened to be similar, e.g., a Democratic budget bill generally includes higher spending than a Republican one.} We now describe the timing of the discovery game, noting variations as they arise:
\begin{enumerate}
	\item The principal chooses whether to discover $v_1$, $v_2$, or neither.\footnote{In Section \ref{subsec:both}, we also consider the possibility of discovering both values.}
	\item The value $v_i$ of the chosen project, if any, is publicly revealed.\footnote{In Section \ref{subsec:noisy}, we consider the case where a noisy public signal of $v_i$ is observed, so that $v_i$ is not known with certainty.} The agent forms a posterior belief $G$ about $v$ by using Bayes's Rule.
	\item The principal makes a proposal by selecting a subset of projects $S$. This subset need not contain the discovered project (and also need not exclude it).
	\item The agent approves the proposal so long as $\sum_{i \in S} \mathbb{E}_G[v_i] \geq 0$.
	\item The principal receives $\sum_{i \in S} w_i$ if her proposal is approved, and 0 otherwise.\footnote{In Section \ref{subsec:benchmark}, we consider other variations of the principal's utility.}
\end{enumerate}
This formal description allows us to make precise the language of trial balloons and consensus projects from the introduction. Using a \textit{trial balloon} means discovering a project $i$ with $\mu_i < \mu_j$ because project $i$ is ex-ante less favored by the agent, as well as $\sigma_i \geq \sigma_j$ because the trial balloon's value is more uncertain. We refer to a project $i$ with $\mu_i < \mu_j$ without restriction on its variance as \textit{disfavored} and call project $j$ a \textit{consensus project}. This description is a slight misuse of notation because it may be that $\mu_j < 0$ and both projects are ex-ante disfavored, but we will focus especially on the case where $\mu_i < 0 < \mu_j$, so that the consensus project is, as the name suggests, ex-ante open to approval by the agent. This case allows us to clearly contrast the spoiling and improvement effects discussed in the introduction.

\section{Single-Project Discovery}
\label{sec:bvNorm}
In this section, we fully characterize the cases in which discovering the value of a trial balloon is best for the principal. We also describe optimal discovery of disfavored projects even when their variance is lower than that of the favored alternative. The principal's decision in the single-project discovery setting is fully captured by the spoiling versus improvement trade-off. Both effects will favor trial balloons, a result we formalize in Theorem \ref{thm:trialBalloon}, though as we show in that result it is still important to consider the settings in which the principal prefers no discovery at all. However, even when a disfavored project has a lower variance, and thus is less effective at improving the agent's beliefs, we show in Theorem \ref{thm:optSingle} that its lower chance of spoiling may still dominate for an intermediate range of favorabilities. These main results also provide two-part conditions which describe the effects of increasing prior precision depending on whether the sender places low or high weight on the disfavored project. We defer full proofs of all results to Appendix A, and provide intuitions in the main text.

\subsection{Benchmark Cases}
\label{subsec:benchmark}
Before characterizing trial balloon optimality, we first introduce some results for simple benchmark cases that restrict the prior distribution of $v$ or the principal's preferences. These insights are useful steps in proving more general results, but may also be of independent interest to model particular settings.

To build intuition, assume $\rho = 1$, so that the two projects are perfectly correlated. This choice abstracts away from the comparison of trial balloons versus building consensus and instead focuses on the simplest trade-off between improving beliefs and spoiling an ex-ante favored project:
\begin{prop}[Perfectly Correlated Projects]
	\label{prop:fullCor}
	When $\rho = 1$, then discovering one project is the same as discovering both, and the following is optimal:
	\begin{enumerate}
		\item If $\mu_1 + \mu_2 \geq 0$: no discovery.
		\item Else if $\mu_1 < 0$ and $\mu_2 < 0$: full discovery.
		\item Else $\mu_i = -\mu < 0 \leq c\mu = \mu_j$ for $c \in [0, 1)$. Then there is a cutoff $c^* \in [0, 1]$ such that either
		\begin{enumerate}
			\item Very disfavored project 1: $c < c^*$, no discovery.
			\item Mildly disfavored project 1: $c \geq c^*$, full discovery.
		\end{enumerate}
	\end{enumerate}
\end{prop}
Using the structure of the bivariate normal distribution, we obtain a closed-form expression for $c^*$, including conditions for the cases when $c^* = 0$ so that there is always full discovery or $c^* = 1$ so that there is never full discovery. For small $c$, the principal is unwilling to risk spoiling the sure thing, but as $c$ grows larger the possibility of improvement becomes more tempting. To showcase the effect of variance, consider the case of $w_1 = w_2 = 1/2$, where we obtain a simple expression, $c^* = 1/(2 - \sigma_1/\sigma_2)$. As overall variance increases (captured by increasing both $\sigma_1$ and $\sigma_2$), the probabilities of spoiling and improvement both increase. However, even in this simple case the net effect of increased variance is uncertain.

To highlight the role of project variance in improvement, we can modify the principal's preferences so that she \textit{only} gets positive utility if both projects are approved. In this case, straightforward computation of the explicit form of the agent's posterior is all that is needed for the result:
\begin{prop}[Complementary Projects]
	\label{prop:bundle}
	Assume the principal only cares about the probability of obtaining the grand bundle, and $\mu_1 + \mu_2 < 0$ so that it cannot be obtained with certainty. When discovering only one project, the higher-variance one is always preferred. Discovering both projects is strictly better than discovering either individually.
\end{prop}
Intuitively, there is no issue of spoiling, because preserving a single ex-ante favored project is useless for a principal who acts as if the projects are perfectly complementary. Thus all that matters is the possibility of an above-mean realization that raises the agent's posterior belief, which is affected only by the variances of the distributions from which discovered values are drawn. This case captures settings where the principal's agenda control is more limited. For example, a firm advertising an open position can choose which aspects of the job to feature in an ad campaign (and can ``commit" to the campaign by rolling out ads before beginning to interview workers) but is only interested in attracting workers who are willing to do all the projects the firm considers part of the job. Our result suggests that firms benefit from focusing on a project where potential employees' opinions are widely dispersed regardless of the average candidate's opinion.

We also consider the case of perfectly substitutable projects, requiring that both projects are ex-ante disfavored to ensure there is some discovery:
\begin{prop}[Substitutable Projects]
	\label{prop:subs}
	Assume the principal cares only about the probability of having at least one project approved. Let $\mu_1 < \mu_2 < 0$ so that approval cannot be obtained with certainty. Discovering the project $i$ satisfying $\sigma_i/\mu_i < \sigma_j/\mu_j$ is better than discovering project $j$, and discovering both projects is strictly better than discovering either individually.
\end{prop}
As in the previous proposition, there is once again no possibility of spoiling, because the presence of any ex-ante favored project in this environment would trivially lead to no discovery. Thus it is still the ratio of variances that determines which project is discovered, with the modification that a lower-variance project may be discovered if its low mean makes it most likely to yield a positive draw. Full discovery is better than discovery of any individual project because, given both projects have negative means, the law of iterated expectations ensures that $\mathbb{E}[\mathbb{E}[\mu_i | v_j]] = \mu_i < 0$\textemdash the average posterior mean of an undiscovered project is unchanged. Thus the principal does better in expectation by discovering that project as well, and chooses to do so by discovering both projects at once if possible. Although this case is less prominent in the real world than that of perfect complementarity, we believe it may still capture some organizational settings. For example, a middle manager may need an executive's approval to hire an additional employee for a team, but only has one position to fill; our result pins down the trade-off between hiring a more favored candidate (as captured by the mean) and hiring one with a better chance of a surprising performance (as captured by the variance).

\subsection{Optimality of Trial Balloons}
\label{subsec:highVar}
We now turn to the general case of $\rho \in (0, 1)$ and a principal with weakly positive weights $w = (w_1, w_2)$ for the two projects satisfying $w_1 + w_2 = 1$. Our first result characterizes optimal single-project discovery when no project is ex-ante open to approval by the agent:
\begin{prop}[Discovery with No Favored Projects]
	\label{prop:disfavor}
	Let $\mu_1 < \mu_2 < 0$, and let $w_1 = 1/2$. Fixing $\sigma_2$, there exists $\underline{\sigma} \in (\sigma_2, (\mu_1/\mu_2) \, \sigma_2)$ such that discovering project 1 is best if and only if $\sigma_1 > \underline{\sigma}$. Otherwise, discovering project 2 is best.
\end{prop}
Imposing the condition $w_1 = w_2 = 1/2$ on the relative weights allows us to decompose the principal's utility into a ``complements term" and a ``substitutes term," and invoke the results of Propositions \ref{prop:bundle} and \ref{prop:subs} to establish which discovery rule maximizes each of those terms. The trial balloon under-performs relative to the complements case but over-performs relative to the substitutes case. Although it is clear that increasing $w_i$ favors discovery of project $i$, we are unable to verify that the utilities of these discovery rules remain single-crossing. However, thorough simulation\footnote{We use a 1,000 point grid for each of $\rho$, $w_1$, $\mu_1/\mu_2$, and $\sigma_1/\sigma_2$.} suggests that single-crossing is preserved for all weights, suggesting that the relationship above would hold for a decreasing continuous cutoff function $\underline{\sigma}(w_1)$.

To sharpen the distinction between trial balloons and consensus projects, we turn to the case where the former is ex-ante disfavored by the agent, whereas the latter is ex-ante open to approval. 
\begin{theorem}[Optimal Trial Balloons]
	\label{thm:trialBalloon}
	Let $\mu_1 = -\mu < 0 \leq c \mu = \mu_2$ for $c \in [0, 1)$ and let $\sigma_1 \geq \sigma_2$, so that project 1 is a trial balloon.\\
	For any $\mu$, there is a cutoff $\underline{c}(\mu)$ such that it is optimal to discover neither project when $c < \underline{c}(\mu)$ and to discover project 1 when $c > \underline{c}(\mu)$.\\
	The function $\underline{c}(\mu)$ is continuous and has the following properties:
	\begin{enumerate}
		\item If $w_1 = 1/2$, then
		$$\underline{c}(\mu) = c^* := \frac{\rho \sigma_2}{\sigma_1 + 2 \rho \sigma_2} \qquad \text{for all } \mu.$$
		\item If $w_1 > 1/2$, there is $\mu^* \in \mathbb{R}_+$ such that $\underline{c}(\mu) = 0$ for all $\mu \in [0, \mu^*]$. For $\mu > \mu^*$, the cutoff $\underline{c}(\mu)$ is strictly increasing with
		$\lim_{\mu \rightarrow \infty} \underline{c}(\mu) = c^*.$
		\item If $w_1 < 1/2$, there is $\mu^* \in \mathbb{R}_+$ such that $\underline{c}(\mu) = 1$ for all $\mu \in [0, \mu^*]$. For $\mu > \mu^*$, the cutoff $\underline{c}(\mu)$ is strictly decreasing with
		$\lim_{\mu \rightarrow \infty} \underline{c}(\mu) = c^*.$
	\end{enumerate}
\end{theorem}
Given the variance condition, the trial balloon is always better than building consensus because it is both less likely to spoil the favored project and more likely to deliver an above-mean realization that attains the grand bundle. The remaining question is whether the trial balloon is better than no discovery at all.\footnote{Although we focus on optimality of trial balloons, the proof of this result can also be used to compare discovery of the consensus project to no discovery.} A large value of $c$, meaning that the trial balloon's disfavor is about equal in magnitude to the consensus project's favor, encourages discovery by making spoiling less likely. However, there is a subtle interaction between the minimal such $c$, the precision of the agent's prior (as captured by taking $\mu$ to be large relative to a fixed ratio of $\sigma_1$ and $\sigma_2$), and the weight $w_1$ assigned to the trial balloon. When the principal places high importance on the trial balloon, spoiling is not especially concerning compared to obtaining approval for the trial balloon via improvement. Thus, low precision encourages discovery by making both effects strong, and with especially low precision trial balloons are in fact optimal at all values of $c$. This low cutoff comes with a cost, though: for $c < c^*$, improvement vanishes more quickly than spoiling as precision increases. Thus for a fixed $c$ in this range, no discovery will eventually dominate discovering the trial balloon. Thus the cutoff when $w_1 > 1/2$ must increase in $\mu$. The effect when $w_1 < 1/2$, so the principal places low weight on the trial balloon, is the opposite: the cutoff for trial ballon discovery begins above $c^*$, so improvement vanishes more slowly than spoiling as $\mu$ increases, and for any $c > c^*$ the trial balloon eventually dominates. Figure \ref{fig:trialBalloon} provides an example of the cutoff in each of these two cases.

\begin{figure}
	\centering
	\includegraphics[width=0.45\hsize]{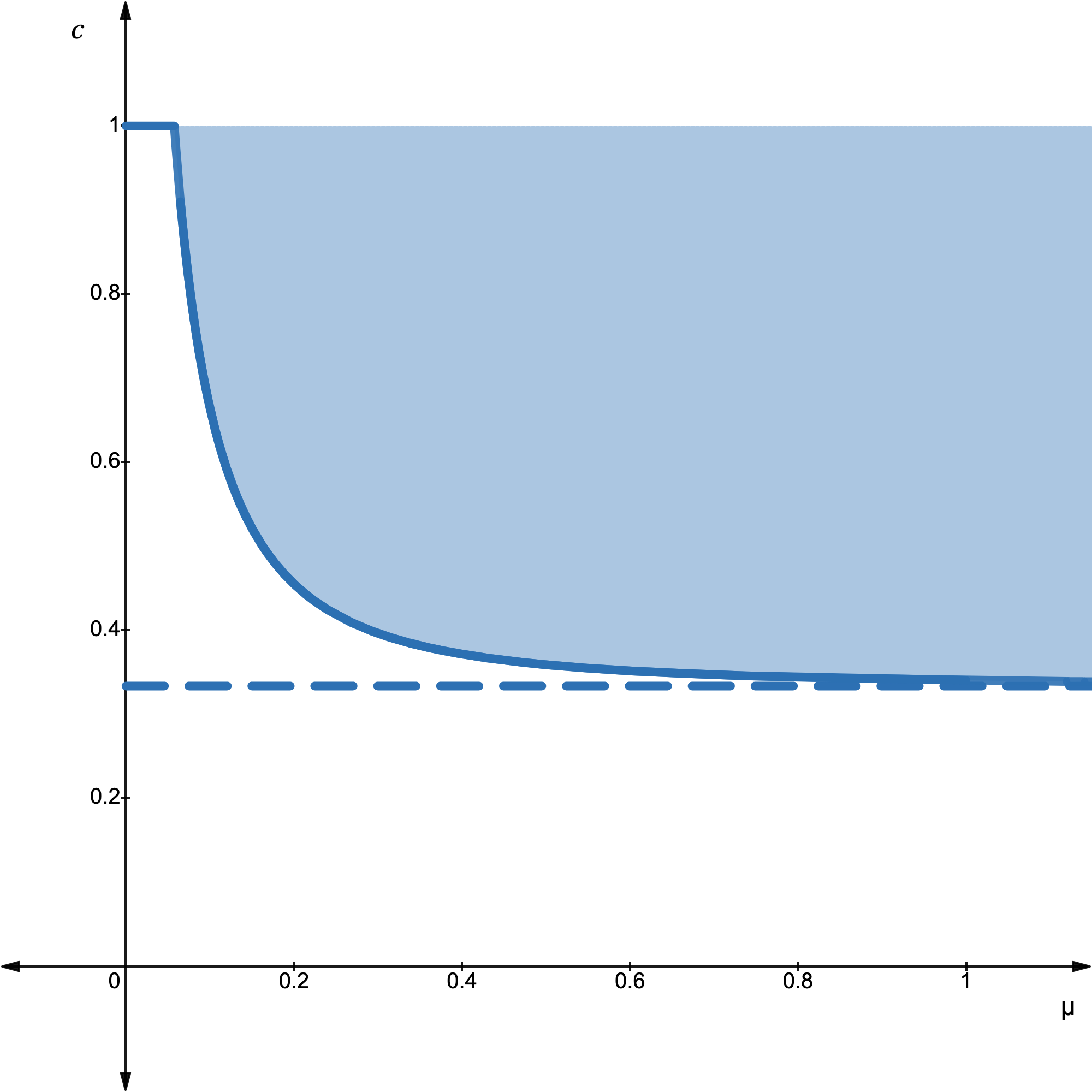}
	\hfill
	\includegraphics[width=0.45\hsize]{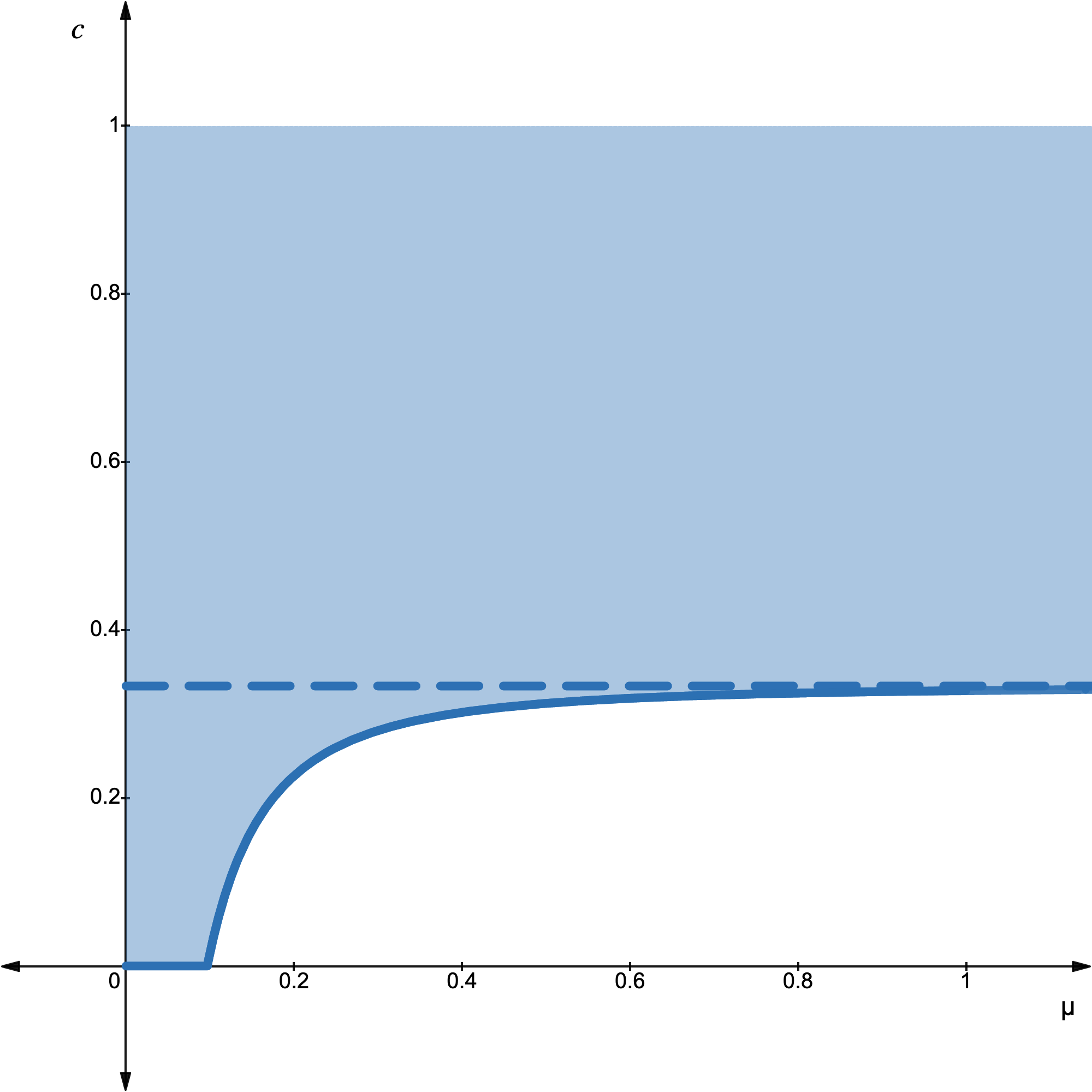}
	\caption{\small An example of Theorem \ref{thm:trialBalloon}: in the left panel $w_1 < 1/2$ and in the right panel $w_1 > 1/2$. The dashed line marks the limit value $c^*$, and the trial balloon is used when $(\mu, c)$ is in the blue shaded region.}
	\label{fig:trialBalloon}
\end{figure}
In the political context, where we expect priors about a presidential administration's policies to grow more precise later in that president's term, our results make an important testable prediction about the prevalence of trial balloons. If a trial balloon topic is relatively unimportant to the president's priorities, we would expect more trial balloons later in the term; if instead the trial balloon topic is central to the administration's agenda, we would expect early trial balloons followed by convergence to consensus-building proposals made without any discovery. Tracking administration priorities (e.g., through text analysis of press releases) and perceived support for various projects (e.g., through recurring poll questions) would allow us to test this prediction in the data.

\subsection{Disfavored Projects with Lower Variance}
\label{subsec:precise}
When $\sigma_1 < \sigma_2$, the disfavored project is no longer a trial balloon; it continues to have a lower risk of spoiling but now has a lower chance of improvement as well. This trade-off alters the choice of optimal discovery rule and carves out room for building consensus. Although the comparison of either project to no discovery is unchanged, comparing the two projects to each other requires careful consideration of the rates at which spoiling and improvement vanish. In the theorem below, we use the limit thresholds at which each project dominates no discovery to make the relevant comparisons.
\begin{theorem}[Optimal Single-Project Discovery]
	\label{thm:optSingle}
	Let $\mu_1 = -\mu < 0 \leq c \mu = \mu_2$ for $c \in [0, 1)$ and let $\sigma_1 < \sigma_2$ so that project 1 is disfavored but not a trial balloon.\\
	Define $c_i(\mu)$ as the cutoff above which project $i$ is preferred to no discovery. When $w_1 = 1/2$, the cutoffs are constant in $\mu$, with
	$$c_2(\mu) = c^{**} := \frac{\sigma_2}{\rho \sigma_1 + 2\sigma_2} > \frac{\rho \sigma_2}{\sigma_1 + 2 \rho \sigma_2} =: c^* = c_1(\mu).$$
	
	The value of $\mu$ divides optimal discovery into three regions: \footnote{The boundaries of these regions differ depending on the value of $w_1$. We include a precise characterization in Appendix A, and focus here on the behavior of the optimal discovery rule.}
	\begin{enumerate}
		\item Low $\mu$: no discovery if $c < \min \left\{c_1(\mu), c_2(\mu)\right\}$ and discovery of an unknown project otherwise.
		\item High $\mu$: the cutoff $c_1(\mu)$ lies strictly below $c_2(\mu)$. There is no discovery if $c < c_1(\mu)$, discovery of project 1 if $c \in (c, c_1(\mu))$, and discovery of an unknown project if $c > c_1(\mu)$.
		\item Limit $\mu$: no discovery if $c < c^*$, discovery of project 1 if $c \in (c^*, c^{**})$, and discovery of project 2 if $c > c^{**}$.
	\end{enumerate}
\end{theorem}
\begin{figure}
	\centering
	\includegraphics[width=0.45\hsize]{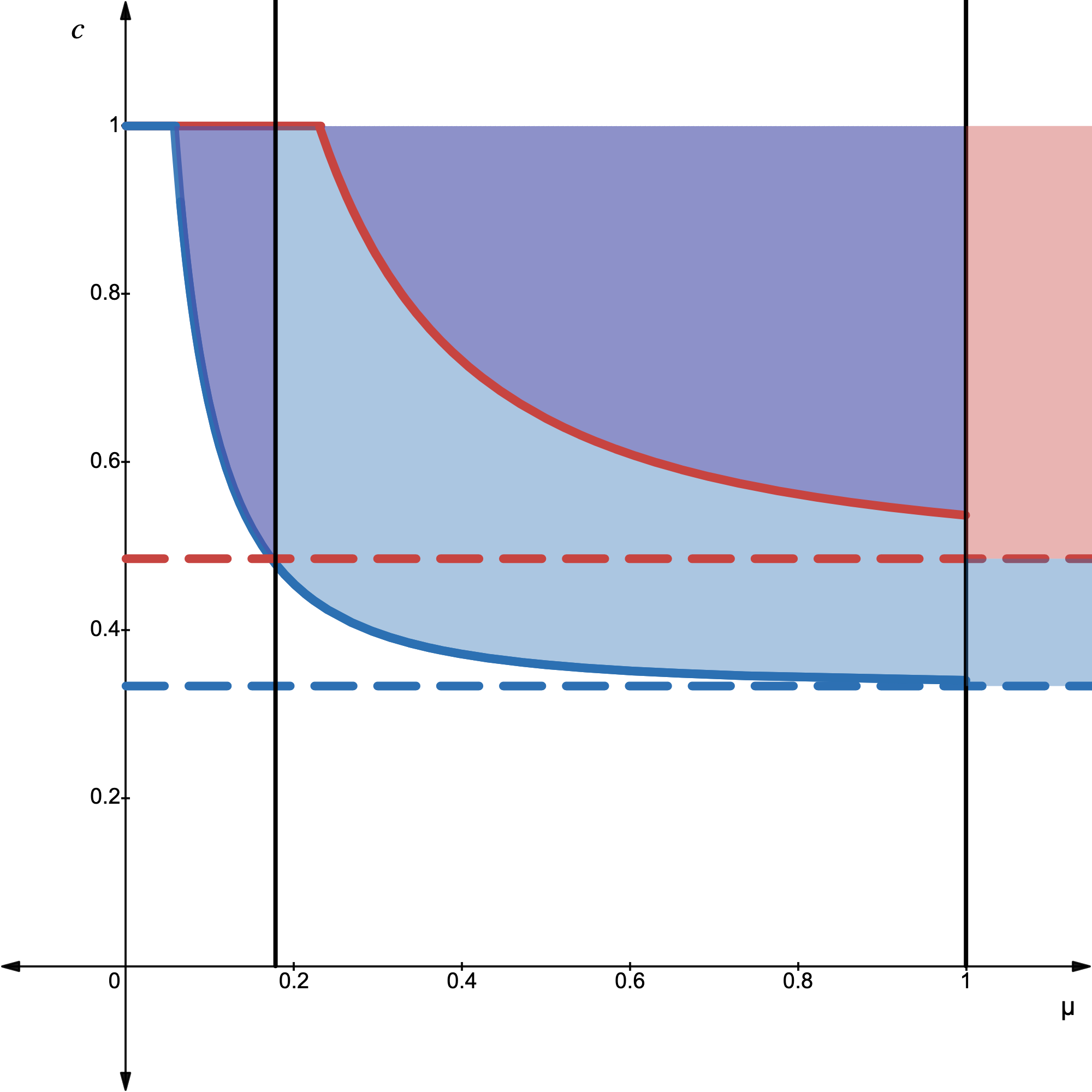}
	\hfill
	\includegraphics[width=0.45\hsize]{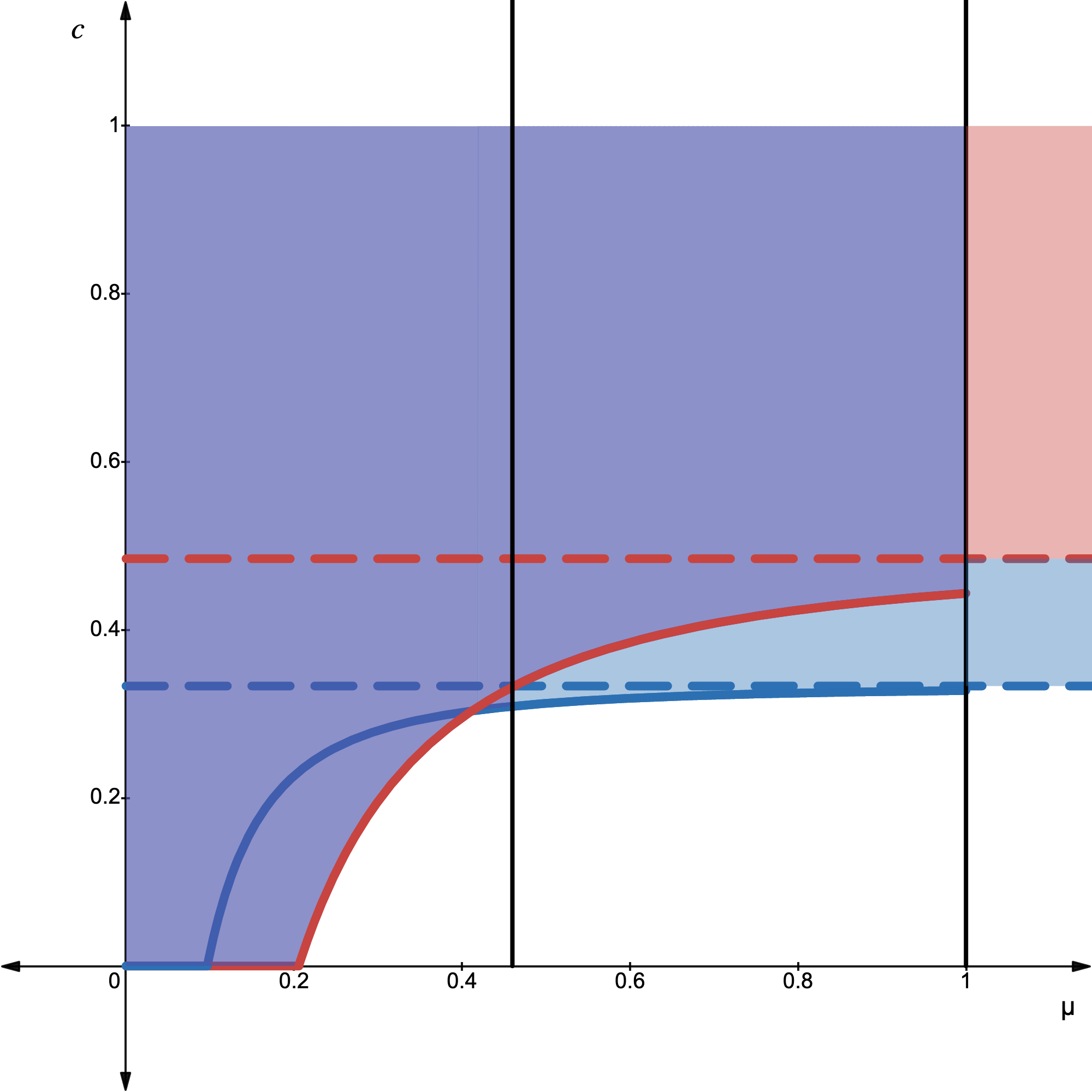}
	\caption{An example of Theorem \ref{thm:optSingle}: in the left panel $w_1 < 1/2$ and in the right panel $w_1 > 1/2$. In blue, project 1 is discovered; in red, project 2 is discovered; in purple, an unknown project is discovered. The black vertical lines separate the three cases, with $\mu > 1$ standing in for the limit case.}
	\label{fig:optSingle}
\end{figure}
In the limit, the interval $[0, 1)$ of possible values of $c$ can be partitioned into three sub-intervals such that the optimal discovery rule is constant in each sub-interval. However, this property does not hold outside of the limit. Depending on the model parameters, numerical simulations reveal that a variety of patterns for optimal discovery are possible. For example, under the parameters used to produce the right panel of Figure 2,\footnote{Specifically, we set $\sigma_1 = 1/20$, $\sigma_2 = 1/5$, $\rho = 1/4$, and $w_1 = 2/3$.} computing the principal's utility shows that the regions where it is optimal to discover project 1 are disconnected.

The clearest conclusion of Theorem \ref{thm:optSingle} is that, for large enough $\mu$, discovering the disfavored project is best for intermediate values of $c$. Furthermore, if $w_1 < 1/2$ then this intermediate range narrows as precision increases, so that discovery in general becomes more common but disfavored projects as a share of discovered projects become less common. If instead $w_1 > 1/2$, the opposite holds; discovery becomes less common, but more of the discovered projects are disfavored.\footnote{Lemma \ref{lm:compStat} in Appendix A proves this monotonicity property in the limit $\mu$ region. However, monotonicity does not hold away from the limit for all parameter values, though it does arise in a high percentage of simulations.} These predictions contrast with those of Theorem \ref{thm:trialBalloon}, and highlight the difference between the trial balloon setting and one in which the consensus project is more effective in improving the agent's beliefs. In the trial-balloon case, the changing range of optimality for the disfavored project depended only on whether spoiling or improvement vanished more quickly. Here, it depends also on the rate at which they vanish compared to the consensus project. When $w_1 < 1/2$, vanishing spoiling enlarges the region where discovering some project is optimal. Because spoiling vanishes more rapidly for the disfavored project, the optimality of discovering it for values ever closer to $c^*$ is overcome by the optimality of building consensus for values ever closer to $c^{**}$. When instead $w_1 > 1/2$ and vanishing improvement reduces the region where discovering some project is optimal, faster-vanishing improvement for the disfavored project leaves room for the optimal discovery region to grow as $c_2(\mu)$ more slowly approaches $c^{**}$.

To see the practical implications of the result, consider a politician that chooses not between a trial balloon and a popular low-variance policy, but between a polarized policy and an appealing generality. The polarized policy may be pork-barrel spending that is unpopular on average because only voters in a given district approve; given the definite nature of the spending, there is little chance public opinion will change. The appealing generality may be a promise to cut taxes; the average voter likes the idea, but depending on the details approval will likely change. Suppose further that the politician is self-interested and values a bill containing pork-barrel spending more than one with tax cuts. Then, as the public forms a more precise prior, the politician will be more likely to directly propose legislation without first making public commitments (less discovery overall). However, when the politician does speak in order to gauge the audience's response, they will be more likely to focus on the pork-barrel spending (relatively more discovery of the disfavored project). This prediction draws a sharp distinction between optimal behavior when the disfavored project's variance is lower rather than higher, and could be tested in a similar way to Theorem \ref{thm:trialBalloon}.

\section{Extensions}
\label{sec:ext}
The previous section captures the intuitive trade-off between spoiling and improvement that drives the optimality of trial balloons, but makes stark assumptions about the discovery rules available to the principal. In this section, we relax first the assumption of single-dimensional discovery and then that of conclusive evidence to show how our qualitative conclusions are preserved in the face of these changes.

\subsection{Discovering Both Projects}
\label{subsec:both}
We begin by considering the possibility of full discovery, where in addition to choosing between no discovery and discovery of a single project, the principal can choose to publicly reveal the value of both projects. For simplicity, we set $w_1 = w_2 = 1/2$, and omit the comparison of single-project discovery rules to no discovery:

\begin{theorem}[Multi-Project Discovery]
	\label{thm:both}
	Optimal discovery, allowing for discovering both projects, is as follows:
	\begin{enumerate}
		\item If $\mu_1 < 0$ and $\mu_2 < 0$, discover both projects.
		\item Else $\mu_1 = -\mu < 0 \leq c \mu = \mu_2$ for $c \in [0, 1)$. Assume $\sigma_1 \geq \sigma_2$ so that project 1 is a trial balloon. Then, there are values $0 < c_\ell < c_h < 1$ such that
		\begin{enumerate}
			\item for any $c \in [0, c_\ell)$ there exists $\underline{\mu}$ such that discovering project 1 is better than discovering both projects for all $\mu > \underline{\mu}$.
			\item for any $c \in [c_h, 1)$ there exists
			$\bar{\mu}$ such that discovering project 1 is worse than discovering both projects for all $\mu > \bar{\mu}$.
		\end{enumerate}
		There is also $c_\text{no} \in (0, 1)$ such that
		\begin{enumerate}[start=3]
			\item for any $c \in [0, c_\text{no})$ there exists a continuous function $\mu_\text{no}(c)$ such that full discovery is better than no discovery if and only if $\mu \in [0, \mu_\text{no}(c))$.
			\item for any $c > c_\text{no}$, full discovery is better than no discovery for all $\mu$.
		\end{enumerate}
	\end{enumerate}
\end{theorem}
We use numerical simulations similar to those in Proposition \ref{prop:disfavor} to verify that the gap between $c_\ell$ and $c_h$ is closed at some intermediate $c_\text{mid}$, below which case (2a) occurs and above which case (2b) occurs. As with the non-limit cases of Theorem \ref{thm:optSingle}, we can find counterexamples for finite $\mu$ where the dominance regions for each discovery rule are not connected. Finally, numerical values of $c_\text{no}$ lie both above and below $c_\text{mid}$ depending on the other parameters.

Case (1) here parallels Proposition \ref{prop:disfavor} in applying the results of a complements-and-substitutes decomposition of the principal's utility. However, rather than facing a potential trade-off between the two terms which leads to a ``compromise" cut-off, in this case the greater variance coming from the ability to discover both projects dominates any single-project discovery rule. Combining case (2) with Theorem \ref{thm:trialBalloon} leads to a similar intermediate region of dominance for the trial balloon as the disfavored project in Theorem  \ref{thm:optSingle}. We can also obtain the same limiting comparative statics with respect to increased precision, with increased precision shrinking the dominance region when $w_1 < 1/2$ and expanding it when $w_1 > 1/2$ even as the behavior of the overall discovery region is opposite. Although the limitation to single-project discovery is often natural, this result reinforces the conclusions of Theorem \ref{thm:optSingle} by making clear that depend only on the magnitude and convergence rate of the spoiling and improvement probabilities as functions of some arbitrary discovery rule.

\subsection{Sequential Discovery}
\label{subsec:seq}
Once the principal is allowed the possibility of discovering both projects at once, it is natural to consider sequential discovery in comparison to the one-shot discovery rule described above. This setting is challenging to work with, but we can establish several basic facts. First, the principal never discovers both projects at once, and her utility is trivially higher than in the one-shot case. The results of Theorem \ref{thm:trialBalloon}, appropriately modified to use the parameters of the posterior belief after the first discovery step, determine whether a second project is discovered. Finally, the regions where a project will be discovered are larger than in either of Theorems \ref{thm:trialBalloon} and \ref{thm:both}, because the ability to discover a second project offers the principal a chance at  ``redemption" by discovering a second project afterward. 

Two obstacles remain to fully characterize the sequential model: a tight characterization of when the first discovery step is undertaken, and a description of the order of optimal discovery. We leave these for future work.

\subsection{Noisy Discovery}
\label{subsec:noisy}
The conclusive evidence setting may be attractive for some examples, such as an engineer choosing which project to prototype before presenting to management, but it is also a stark restriction on the principal's ability to discover information. In this section we allow the principal to discover a noisy signal of a project's value rather than discovering the value with certainty.

\begin{prop}[Discovery with Exogenous Noise]
	\label{prop:noise}
	Let discovery of issue $i$ generate a noisy signal $s_i$ of its value with distribution $s_i \sim N(v_i, \tau_i^2)$. Then the results of Section \ref{sec:bvNorm} apply unchanged after substituting the distribution 
	$$N\bigg(\mu_i, \frac{\sigma_i^4}{\tau_i^2 + \sigma_i^2}\bigg)$$
	for the project $i$ which the principal chooses to discover and 
	$$N\bigg(\mu_j + \rho \frac{\sigma_j}{\sqrt{\sigma_i^2 + \tau_i^2}} \, (s_i - \mu_i), (1 - \rho)^2 \sigma_j^2\bigg)$$
	for the posterior belief about project $j$, conditional on discovery of project $i$.
\end{prop}
A noisy signal weakens both the spoiling and improvement effects for project $i$ by reducing the amount by which discovery shifts the agent's beliefs. Thus the effect on optimal discovery is uncertain. However, the presence of exogenous noise does not alter the principal's fundamental trade-off; she simply accounts for the signal's noise and considers a problem in the style of Section \ref{sec:bvNorm} with an appropriately modified prior belief. In general, making the signal of a project's value noisy both makes the principal less likely to discover it and weakly decreases her utility from the optimal discovery rule. This result suggests that future work may benefit from considering the possibility for the principal to reduce the signal's variance at some effort cost.

\subsection{Discovery with Many Projects}
\label{subsec:Nproj}
Thus far we have abstracted away from the motivating issue of a complex proposal by considering a case with only two correlated projects, rather than some large but finite number of projects $N$. This choice is made for tractability, and we can view the two projects in our model as standing in for the types of projects a principal might have access to in a richer setting. Nevertheless, we can still draw some limited conclusions with an explicit model of $N > 2$ projects. 

Let the marginal distribution of the agent's value for project $i \in \left\{1, 2,..., N\right\}$ be $v_i \sim N(\mu_i, \sigma_i^2)$, and assume that the joint distribution of values is an $N$-variate normal with equal correlation $\rho$ between any pair of projects. Then Propositions \ref{prop:bundle} and \ref{prop:subs}, describing perfect complements and perfect substitutes, extend without modification. Specifically, if the principal maximizes the probability of having the grand bundle approved, she always discovers the project with highest variance. If instead she maximizes the probability of having at least one project approved, she discovers any project $i$ satisfying $\sigma_i/\mu_i = \min_{k \in \left\{1,...,N\right\}} \sigma_k/\mu_k$. In the latter case, she also always prefers discovering more projects to discovering fewer.

Obtaining parallel results to Theorems \ref{thm:trialBalloon} and \ref{thm:optSingle} is more challenging; rather than comparing spoiling and improvement effects across two projects, we now have to make the comparison across $N$ projects. These pairwise comparisons are further complicated by the ability to spoil or improve multiple projects at once, rather than only spoiling the single favored project or improving the single disfavored project. To accommodate these issues and preserve our focus on discovery of trial balloons, the following result allows for only one negative-mean project:
\begin{prop}[Optimal Single-Project Discovery with Many Projects]
	\label{prop:Nballoon}$\text{} $\\
	Let $\mu_1 = -\mu < 0 \leq c_j \mu$ with $c_j \geq 0$ for all projects $j \in \left\{2,...,N\right\}$, and let $\sum_{j \in \left\{2,...,N\right\}} c_j \leq 1$. Then the following analogues of previous theorems hold:
	\begin{itemize}
		\item (Analogue of Theorem \ref{thm:trialBalloon}, part 1): For any project $k$, choosing one $c_j$ and fixing the others, there is a continuous function $\underline{c}_j^k(\mu, w_1)$ such that discovering that project is better than no discovery if and only if $c_j > \underline{c}_j^k(\mu, w_1)$. This function is decreasing in $w_1$ for any $\mu$.
		\item (Analogue of Theorem \ref{thm:trialBalloon}, part 2): If $\sigma_1 = \max_{k \in N} \sigma_k$, then discovering project 1 is always better than discovering any other project. 
		\item (Analogue of Theorem \ref{thm:optSingle}): If $\sigma_i < \max_{k \in N} \sigma_k$, then discovering project 1 has a weaker spoiling effect than discovering project $k$, but need not have a stronger improvement effect.
	\end{itemize}
\end{prop}
Theorem \ref{thm:trialBalloon} extends almost unchanged in this setting; the main difference is that we can no longer analytically sign the behavior of the cutoff $\underline{c}_j^k$ with respect to increasing prior precision, though we can still see that the cutoff decreases as more weight is placed on the trial balloon. The expressions for the limit values of each $\underline{c}_j^k$, analogous to $c^*$ or $c^{**}$ in Theorem \ref{thm:optSingle}, cannot be written in fully closed form. Thus the conclusions of Theorem \ref{thm:optSingle}, which rely on the ordering of those cutoffs, are almost entirely lost; in fact we cannot even be sure that a lower-variance disfavored project does not maintain a stronger improvement effect. These limitations emphasize the complexity of discovery in high-dimensional environments and suggest that a low-dimensional heuristic, such as the spoiling versus improvement trade-off we identify, may serve as a reasonable guide for real-world decisions.

\section{Conclusion}
\label{sec:conc}
Real-world policy proposals bundle many correlated projects, and numerous constraints\textemdash time, money spent, and attention, to name a few\textemdash may restrict a proposer's ability to communicate about the proposal as a whole. To study the choice of which project to emphasize, we present a stylized model with a two-project proposal space where a principal decides whether to publicly discover information about an agent's unknown value in order to gain the agent's approval for their final proposal. In our baseline setting, where the principal can only discover one of the two project values and places additive weight on each project, we characterized the optimality of trial balloons\textemdash ex-ante disfavored projects with higher variance than the favored alternative. The optimal discovery rule is driven by a trade-off between spoiling any ex-ante favored projects and improving the agent's beliefs to get more projects approved. We highlight the interplay between the relative weights, magnitudes, and variances of the two projects, which lead to subtle results. Trial balloons, and discovery of any project in general, become more common with increasing prior precision if and only if the weight on a disfavored project is lower than that on a favored one. However, the share of this region taken up by a lower-variance disfavored project has the opposite sign. If it has low importance, then the disfavored project takes up a smaller share of the growing discovery region; with high importance, it takes up a larger share of the shrinking discovery region. Even when the principal has access to richer discovery rules, such as full or noisy discovery, disfavored projects benefit from the ability to explore larger bundles without foreclosing the possibility of a safe proposal if that project fails. Our setting captures only some of the many factors driving the adoption of trial balloons in political and organizational settings, and there remains work to do in tailoring the broad strokes of this framework to particular settings and the discovery rules that fit them.

\bibliographystyle{ecta}
\bibliography{balloons_v2.bib}

\newpage

\section*{Appendix A: Proofs}
\begin{proof}[Proof of Proposition \ref{prop:fullCor}] $\text{} $
	
	The equivalence of discovering one project and discovering both follows automatically from perfect correlation. We thus refer to discovering both projects throughout, and compare full discover to no discovery only.
	
	If $\mu_1 < 0$ and $\mu_2 < 0$, the payoff without discovery is 0 for sure, whereas discovery may give a positive payoff (if the realized projects have sufficiently high value).
	
	If $\mu_1 + \mu_2 \geq 0$, then the payoff without discovery is 1 for sure, whereas discovery may give a lower payoff (if the realized projects have sufficiently low value).
	
	If $\mu_i = -\mu < 0 \leq c \mu = \mu_j$ for $c \in [0, 1)$, then the payoff without discovery is $w_j$ for sure, as only project $j$ is approved. Consider discovering project $i$; then if $c = 1/(2 - \sigma_1/\sigma_2)$, by symmetry of the marginal distribution of $v_i$ there is an equal probability that 
	$$v_i < \mu_i - (\sigma_i/\sigma_j) \, \mu_j \Rightarrow v_j < \mu_j - \mu_j = 0,$$
	so $\mu_i$ is no longer approved, and  
	$$v_j > \mu_i + (\sigma_i/\sigma_j) \, \mu_j \Rightarrow v_i + v_j = (\mu_i + \mu_j) + (1 + \sigma_j/\sigma_i) \, \mu_i > 0.$$ 
	If $w_j = 1/2$, the expected payoff is therefore the same as no discovery.
	
	If instead $w_j < 1/2$, then the expected payoff from discovering project 1 is greater than from no discovery, because the increase in utility from getting 1 instead of $w_j$ exceeds the loss from getting 0 instead of $w_j$. As $c \rightarrow 0$ with $\mu$ held fixed, the probability that $v_j < 0$ approaches $1/2$, whereas the probability that $v_i + v_j > 0$ approaches $\mathbb{P}(v_i > 0) = \Phi(\mu/\sigma_1) := p < 1/2$. We therefore have an expected utility of $(1 - 1/2 - p) \cdot w_j + p \cdot 1$ in the limit. When $w_j > 2p/(2p+1)$, this limit value is less than $w_j$, and so by continuity there is a cutoff $\underline{c} \in (0, 1/(2 - \sigma_1/\sigma_2))$ such that discovering project 1 is best if and only if $c > \underline{c}$. When $w_j < 2p/(2p + 1)$, discovering project 1 is best for all $c$.
	
	Finally, if $w_j > 1/2$, then the expected payoff from discovering project 1 is less than from no discovery. As $c \rightarrow 1$, the probability that $v_j < 0$ approaches $\Phi(-\mu/\sigma_2) := q < 1/2$ whereas the probability that $v_i + v_j > 0$ approaches $1/2$. We therefore have an expected utility of $(1 - 1/2 - q) w_j + 1/2 \cdot 1$ in the limit. When $w_j < 1/(2q + 1)$, this limit value exceeds $w_j$, and so by continuity there is a cutoff $\underline{c} \in (1/(2 - \sigma_1/\sigma_2), 1)$ such that discovering project 1 is best if and only if $c > \underline{c}$. When $w_j > 1/(2q + 1)$, discovering nothing is best for all $c$.
\end{proof}

\begin{proof}[Proof of Proposition \ref{prop:bundle}] $\text{} $
	
	Throughout, let $X$ be a $N(0, 1)$ random variable. Note that upon discovering project $i$ and obtaining realization $v_i$, the agent's posterior belief about $v_j$ is
	\begin{equation*}
		v_j | v_i \sim N\bigg(\mu_j + \rho \,  \frac{\sigma_j}{\sigma_i} \, (v_i - \mu_i), (1 - \rho^2) \, \sigma_j^2\bigg).
	\end{equation*}
	Thus the probability of obtaining the grand bundle is given by
	\begin{equation*}
		\begin{split}
			\mathbb{P}(v_i + \mathbb{E}(v_j | v_i) \geq 0) = & \mathbb{P}\bigg(v_i + \mu_j + \rho \,  \frac{\sigma_j}{\sigma_i} \, (v_i - \mu_i) \geq 0\bigg)
			\\& = \mathbb{P}\bigg(v_i \geq \frac{(\rho \sigma_j/\sigma_i) \, \mu_i - \mu_j}{1 + (\rho \sigma_j/\sigma_i)}\bigg)
			\\& = \mathbb{P}\bigg(X \geq \bigg(\frac{(\rho \sigma_j/\sigma_i) \, \mu_i - \mu_j}{1 + (\rho \sigma_j/\sigma_i)} - \mu_i\bigg)/\sigma_i\bigg)
			\\& = \mathbb{P}\bigg(X \geq -\frac{\mu_i + \mu_j}{\sigma_i + \rho \sigma_j}\bigg).
		\end{split}
	\end{equation*}
	When discovering both projects, the probability of obtaining the grand bundle is
	\begin{equation*}
		\begin{split}
			\mathbb{P}(v_i + v_j \geq 0) = \mathbb{P}\bigg(X \geq -\frac{\mu_i + \mu_j}{\sqrt{\sigma_i^2 + \sigma_j^2 + 2 \rho \sigma_i \sigma_j}}\bigg),
		\end{split}
	\end{equation*}
	where we note that the sum of two correlated normals is normal with mean and variance as given above (i.e., larger variance but same mean as if they were uncorrelated).
	
	It is clear that whenever $\sigma_i > \sigma_j$, discovering only project $i$ is better than discovering only project $j$, and vice-versa. This result follows from noting that the only way of attaining the grand bundle is to draw an above-mean realization from the appropriate marginal, and so the higher-variance marginal distribution is the best choice.
	
	Setting $\sigma_i = 1$ (which is without loss because we can then scale the other parameters accordingly), we can also obtain a condition for discovering both projects to be better than discovering only $i$:
	\begin{equation*}
		\begin{split}
			-\frac{\mu_i + \mu_j}{1 + \rho \sigma_j} & >  - \frac{\mu_i + \mu_j}{\sqrt{1 + \sigma_j^2 + 2\rho \sigma_j}}
			\\ (1 - \rho^2) \, \sigma_j^2 & > 0.
		\end{split}
	\end{equation*}
	This condition holds for all valid $\rho$, so the absence of restrictions on $\sigma_j$, $\mu_i$, or $\mu_j$ means that discovering both is in fact better than discovering either individual project.

\end{proof}

\begin{proof}[Proof of Proposition \ref{prop:subs}] $\text{} $
	
	We first show that, when restricted to only discovering a single project, discovering the project $i$ satisfying $\sigma_i/\mu_i < \sigma_j/\mu_j$ is better than discovering the other project. Using the expressions from the proof of Proposition \ref{prop:bundle}, we can see that upon discovering project 1, the probability of having a project approved is
	\begin{equation*}
		\max \left\{\mathbb{P}\bigg(X \geq -\frac{\mu_1}{\sigma_1}\bigg), \mathbb{P}\bigg(X \geq -\frac{\mu_2}{\rho \sigma_2}\bigg) \right\},
	\end{equation*}
	where $X$ is a standard normal random variable. Similarly, the probability of having a project approved upon discovering project 2 is
	\begin{equation*}
		\max \left\{\mathbb{P}\bigg(X \geq -\frac{\mu_1}{\rho \sigma_1}\bigg), \mathbb{P}\bigg(X \geq -\frac{\mu_2}{\sigma_2}\bigg) \right\}.
	\end{equation*}
	From here we can see that when $\sigma_1 \geq \sigma_2 \, \mu_1/(\rho \mu_2)$, both maxima are given by the first probability, and so discovering project 1 is best. When $\sigma_1 \leq  \sigma_2 \, \rho \mu_1/\mu_2$, both maxima are given by the second probability, and so discovering project 2 is best. In the remaining interval, we compare $\mathbb{P}(X \geq -\mu_1/\sigma_1)$ to $\mathbb{P}(X \geq -\mu_2/\sigma_2)$, giving precisely the cutoff in the proposition.\newline
	
	We now turn to the comparison between discovering both projects and discovering only one. We drop the assumption that $\mu_1 < \mu_2$ and instead without loss of generality we normalize $\sigma_1^2 = 1$  (adjusting $\mu_1$, $\mu_2$, and $\sigma_2^2$ as we wish). 
	
	If only project 1 is discovered, then there is a level $v^*$ such that for all $v_1 \geq v^*$, the principal will get at least one project approved, and below that point they will get neither project approved. To be able to get project $1$ approved, the principal needs $v_1 \geq 0$. To be able to get project $2$ approved, she needs $\mathbb{E}[v_2|v_1] \geq 0$, which we can write as
	$$\mathbb{E}[v_2|v_1] = \mu_2 + \rho \sigma_2 (v_1 - \mu_1) \geq 0 \iff v_1 \geq \frac{\rho \sigma_2 \mu_1 - \mu_2}{\rho \sigma_2},$$
	where we assume $\rho > 0$ so as not to reverse the inequality when dividing by $\rho \sigma_2$.
	
	If $\mu_2 \leq \rho \sigma_2 \mu_1$, then the value of $v_1$ needed for $\mathbb{E}[v_2|v_1] \geq 0$ is weakly positive, so $v^* = 0$. Thus in this case, when discovering one project, the principal gets at least one project approved if and only if $v_1 \geq 0$. When discovering both, the principal gets at least one project approved with strictly higher probability, because even when $v_1 < 0$ it is still possible to realize $v_2 \geq 0$, and in that case the second project is approved.
	
	Assume instead that $\mu_2 > \rho \sigma_2 \mu_1$, so that $v^* < 0$. In this case, if $v_1 \geq 0$ then regardless of whether project 2 is discovered, the principal surely gets project 1 approved. If $v_1 \in [v^*, 0)$, then when the principal does not discover project 2, it is surely approved; if the principal does discover project 2, it may or may not be approved. If $v_1 < v^*$, then when the principal does not discover project 2, neither project is approved; if the principal does discover project 2, it may or may not be approved. Thus we can write the difference between discovering both projects and discovering project 1 only as $\delta(v_1)$, a function of the value $v_1$ of the discovered project 1:
	\begin{equation*}
		\delta(v_1) = \begin{cases}
			\mathbb{P}(v_2 \geq 0 \, | \, v_1) - 0 = 1- \Phi\big(- \frac{\mu_2+ \rho \, \sigma_2 \, (v_1-\mu_1)}{\sigma_2 \sqrt{1-\rho^2}}\big):= 1- \Phi(\lambda(v_1)), & v_1 <v^*\\
			\mathbb{P}(v_2 \geq 0 \, | \, v_1) - 1 = - \Phi(- \frac{\mu_2+ \rho \, \sigma_2 \, (v_1-\mu_1)}{\sigma_2 \sqrt{1-\rho^2}}):= -\Phi(\lambda(v_1)), & v_1 \in (v^*,0)\\
			0 & v_1 \geq 0
		\end{cases}
	\end{equation*}
	We aim to show that, in expectation, $\delta(v_1)$ is positive, so that it is better for the principal to discover both projects than to discover project 1 only.
	\begin{align*}
		\int_{-\infty}^\infty \delta(v) \, dv > 0 \\
		\iff \int_{-\infty}^0 \delta(v) \, dv > 0 \\
		\iff \int_{-\infty}^{v^*} \delta(v) \, dv > - \int_{v^*}^0 \delta(v) \, dv\\
		\iff \int_{-\infty}^{v^*} 1- \Phi(\lambda(v)) \, dv > \int_{v^*}^0 \Phi(\lambda(v)) \, dv\\
		\iff \int_{-\infty}^{v^*} (1- \Phi(\lambda(x))) f_1(x) \, dx > \int_{v^*}^0 \Phi(\lambda(x)) f_1(x) \, dx\\
		\impliedby \int_{2v^*}^{v^*} (1- \Phi(\lambda(x))) \, f_1(x) \, dx > \int_{v^*}^0 \Phi(\lambda(x)) \, f_1(x) \, dx \\
		\impliedby (1- \Phi(\lambda(v^*-x))) \, f_1(v^*-x) > \Phi(\lambda(v^* + x)) \, f_1(v^* + x) \qquad \forall x \in (0, -v^*)\\
		\iff f_1(v^*-x) > f_1(v^* + x) \qquad \forall x \in (0, -v^*) \checkmark
	\end{align*}
	Going line-by-line, the first implication is because $\delta(v) = 0$ for $v \geq 0$, and the second by splitting the integral at $v^*$. The third substitutes in the expression for $\delta(v)$, and the fourth does a change of variables to integrate against a uniform density rather than the density of project 1, $f_1(\cdot)$. The fifth truncates the left-hand integral at $2v^*$, which because the function being integrated is strictly positive, is strictly more restrictive than allowing the integral to go all the way to $-\infty$. The sixth examines the function pointwise, pairing each value $v^* - x$ on the left-hand side with its reflection $v^* + x$ about $v^*$ on the right-hand side. The seventh uses the fact that $v^*$ is precisely the value of $v_1$ for which $\mathbb{E}[v_2 | v_1 = v] = \lambda(v) = 0$, and the standard normal $\Phi(\cdot)$ is symmetric about 0; thus the mass above $\lambda(v^* - x)$ is precisely equal to the mass below $\lambda(v^* + x)$. Finally, because $\mu_1 < 0$ and $\mu_2<0$, it must be that $v^* > \mu_1$ because there must be a positive shock to project 1 for either project to be approved. Therefore, because $f_1(v)$ is strictly decreasing in $|v - \mu_1|$ and $|(v^* - x) - \mu_1| < |(v^* + x) - \mu_1|$ for any $x > 0$, the last inequality does indeed hold for any $x \in (0, -v^*)$. 
\end{proof}

\begin{proof}[Proof of Proposition \ref{prop:disfavor}] $\text{} $
	
	Let $w_1 = 1/2$. Then the principal's utility can be written as 
	\begin{equation*}
		\begin{split}
			\mathbb{P}(\text{both approved}) & + w_1 \, \mathbb{P}(\text{only 1 approved}) + w_2 \, \mathbb{P}(\text{only 2 approved})
			\\& = \mathbb{P}(\text{both}) + \frac{1}{2} \, \mathbb{P}(\text{only 1}) + \frac{1}{2} \, \mathbb{P}(\text{only 2})
			\\& = \mathbb{P}(\text{both}) + \frac{1}{2} \, \mathbb{P}(\text{only 1 or only 2})
			\\& = \mathbb{P}(\text{both}) + \frac{1}{2} \, \big(\mathbb{P}(\text{at least 1}) -  \mathbb{P}(\text{both})\big)
			\\& = \frac{1}{2} \, \mathbb{P}(\text{both}) + \frac{1}{2} \, \mathbb{P}(\text{at least 1}).
		\end{split}
	\end{equation*}
	From Propositions \ref{prop:bundle} and \ref{prop:subs}, we can see that both expressions are increasing more quickly in $\sigma_1$ when discovering project 1 than when discovering project 2. Furthermore, for $\sigma_1 < \sigma_2$ we know both terms are greater for discovering project 2, and for $\sigma_1 > (\mu_1/\mu_2) \, \sigma_2$ we know both terms are greater for discovering project 1. Thus there must be a cutoff $\underline{\sigma} \in (\sigma_2, (\mu_1/\mu_2) \, \sigma_2)$ such that discovering project 1 is best if and only if $\sigma_1 > \underline{\sigma}$.
\end{proof}

\begin{proof}[Proof of Theorem \ref{thm:trialBalloon}] $\text{} $
	
	We begin by decomposing the principal's utility as
	\begin{equation*}
		\begin{split}
			\mathbb{P}(\text{both approved}) & + w_1 \, \mathbb{P}(\text{only 1 approved}) + w_2 \, \mathbb{P}(\text{only 2 approved}) 
			\\ & = w_1 \, \mathbb{P}(\text{only 1 or both}) + w_2 \, \mathbb{P}(\text{only 2 or both})
			\\ & = w_1 \, \mathbb{P}(\text{both}) + w_2 \, \mathbb{P}(\text{belief}(v_2) \geq 0).
		\end{split}
	\end{equation*}
	By $\text{belief}(v_2)$, we refer to the agent's belief about the expected value of $v_2$, which may either be the actual value of $v_2$ (if project 2 was discovered) or the conditional expectation $\mathbb{E}[v_2 | v_1]$ (if project 1 was discovered). The last line comes from noting two facts. First, because $\mu_1 < 0 < \mu_2$, regardless of which project is discovered, if project 1 would be approved alone then surely the grand bundle will be approved.  Second, because $\mu_2 > 0$, then if the grand bundle would be approved then surely project 2 would be approved alone. 
	
	Let $X$ be a standard normal random variable. Using the decomposition of the principal's payoffs, the payoff for discovering only project 1 is
	\begin{align*}
		\Pi^1 &= w_1 \, \mathbb{P}\big(v_1 + \mathbb{E}[v_2 | v_1] \geq 0\big) + w_2 \, \mathbb{P}\big(\mathbb{E}[v_2 | v_1]\geq 0\big)\\
		&= w_1 \, \mathbb{P}\bigg(X > \frac{(1-c)\mu}{\sigma_1+\rho \sigma_2}\bigg) + w_2 \, \mathbb{P}\bigg(X > \frac{-c\mu}{\rho \sigma_2}\bigg) \\
		&= 1 - w_1 \, \Phi\bigg(\frac{(1-c)\mu}{\sigma_1+\rho \sigma_2}\bigg) - w_2 \, \Phi\bigg(\frac{-c\mu}{\rho \sigma_2}\bigg),
	\end{align*}
	whereas the payoff for discovering only project 2 is
	\begin{align*}
		\Pi^2 &= w_1 \, \mathbb{P}\big(v_2 + \mathbb{E}[v_1 | v_2] \geq 0\big) + w_2 \, \mathbb{P}\big(v_2 \geq 0\big)\\
		&= w_1 \, \mathbb{P}\bigg(X > \frac{(1-c)\mu}{\sigma_2 + \rho \sigma_1}\bigg) + \frac{1}{2} \, \mathbb{P}\bigg(X > \frac{-c\mu}{\sigma_2}\bigg) \\
		&= 1 - w_1 \Phi\bigg(\frac{(1-c)\mu}{\sigma_2 + \rho \sigma_1}\bigg) - w_2 \, \Phi\bigg(\frac{-c\mu}{\sigma_2}\bigg).
	\end{align*}
	Because $\rho \in (0, 1)$, it is always the case that $-c\mu/\sigma_2 > -c\mu/(\rho \sigma_2)$, and thus whenever $\sigma_1 \geq \sigma_2$, it is clear that $\Pi^1 \geq \Pi^2$. 
	
	We now compare discovering project 1 only to discovering neither:
	\begin{equation*}
		\begin{split}
			w_2 & < 1- w_1 \, \Phi\bigg(\frac{(1-c)\mu}{\sigma_1+\rho \sigma_2}\bigg) - w_2 \, \Phi\bigg(\frac{-c\mu}{\rho \sigma_2}\bigg)
			\\1 - w_1 & < 1- w_1 \, \Phi\bigg(\frac{(1-c)\mu}{\sigma_1+\rho \sigma_2}\bigg) - (1 - w_1) \, \Phi\bigg(\frac{-c\mu}{\rho \sigma_2}\bigg)
			\\(1 - w_1) \, \Phi\bigg(\frac{-c\mu}{\rho \sigma_2}\bigg) & < w_1 - w_1 \, \Phi\bigg(\frac{(1-c)\mu}{\sigma_1+\rho \sigma_2}\bigg)
			\\\frac{1 - w_1}{w_1} \, \Phi\bigg(\frac{-c\mu}{\rho \sigma_2}\bigg) & < 1 -  \Phi\bigg(\frac{(1-c)\mu}{\sigma_1+\rho \sigma_2}\bigg) = \Phi\bigg(\frac{-(1-c)\mu}{\sigma_1+\rho \sigma_2}\bigg),
		\end{split}
	\end{equation*}
	where we use symmetry of the normal distribution to get the last equality. We observe that the left-hand side is decreasing in $c$ whereas the right-hand side is increasing in $c$; furthermore, the expression is continuous. Therefore we have one of three cases:
	\begin{enumerate}
		\item The left-hand side is weakly greater than the right-hand side at $c = 1$ (and therefore for all $c \in [0, 1)$ as well). This case occurs if and only if 
		$$\Phi\bigg(\frac{-\mu}{\rho \sigma_1}\bigg) > \frac{w_1}{2 \, (1 - w_1)}.$$ 
		Then, for fixed $\mu$, discovering project 1 is worse than discovering nothing for all $c \in [0, 1)$.  
		\item The left-hand side is weakly less than the right-hand side at $c = 0$ (and therefore for all $c \in [0, 1)$ as well). This case occurs if and only if 
		$$\frac{1 - w_1}{2w_1} < \Phi\bigg(\frac{-\mu}{\sigma_1 + \rho \sigma_2}\bigg).$$
		Then, for fixed $\mu$, discovering project 1 is better than discovering nothing for all $c \in [0, 1)$.
		\item There is a single intersection at $\underline{c} \in (0, 1)$; discovering project 1 is better than discovering nothing if and only if $c > \underline{c}$.
	\end{enumerate}
	We start with the knife-edge case $w_1 = 1/2$, where the inequality reduces to
	$$\Phi\bigg(\frac{-c\mu}{\rho \sigma_2}\bigg) < \Phi\bigg(\frac{-(1-c)\mu}{\sigma_1+\rho \sigma_2}\bigg),$$
	which holds whenever
	$$\frac{-(1-c)\mu}{\sigma_1+\rho \sigma_2} > \frac{-c\mu}{\rho \sigma_2}  \iff c > \frac{\rho \sigma_2}{\sigma_1 + 2 \rho \sigma_2} := c^*.$$
	Thus the cutoff is interior and constant for all $\mu$.

	Now we apply the Implicit Function Theorem to the function
	$$H(\mu, c(\mu)) = H(\mu, c) = \frac{1 - w_1}{w_1} \, \Phi\bigg(\frac{-c\mu}{\rho \sigma_2}\bigg) - \Phi\bigg(\frac{-(1-c)\mu}{\sigma_1+\rho \sigma_2}\bigg)$$
	in order to compute the derivative $c'(\mu)$ at the value $\underline{c}$ which produces equality, noting that we are guaranteed continuity of this function because $H$ is continuous. The total derivative is
	\begin{equation*}
		\begin{split}
			\bigg(\frac{1 - w_1}{w_1} \, \frac{-c}{\rho \sigma_2} \, \phi\bigg(\frac{-c\mu}{\rho \sigma_2}\bigg) & - \frac{-(1 - c)}{\sigma_1 + \rho \sigma_2} \, \phi\bigg(\frac{-(1-c)\mu}{\sigma_1+\rho \sigma_2}\bigg)\bigg) \, d\mu 
			\\& + \bigg(\frac{1 - w_1}{w_1} \, \frac{-\mu}{\rho \sigma_2} \, \phi\bigg(\frac{-c\mu}{\rho \sigma_2}\bigg) - \frac{\mu}{\sigma_1 + \rho \sigma_2} \phi\bigg(\frac{-(1-c)\mu}{\sigma_1+\rho \sigma_2}\bigg)\bigg) \, dc,
		\end{split}
	\end{equation*}
	and equating to 0 gives
	$$\frac{dc}{d\mu} = \frac{-\frac{(1 - w_1) \, c}{w_1 \rho \sigma_2} \, \phi\big(\frac{-c\mu}{\rho \sigma_2}\big) + \frac{(1 - c)}{\sigma_1 + \rho \sigma_2} \, \phi\big(\frac{-(1-c)\mu}{\sigma_1+\rho \sigma_2}\big)}{\frac{(1 - w_1) \mu}{w_1 \rho \sigma_2} \, \phi\big(\frac{-c\mu}{\rho \sigma_2}\big) + \frac{\mu}{\sigma_1 + \rho \sigma_2} \phi\big(\frac{-(1-c)\mu}{\sigma_1+\rho \sigma_2}\big)}.$$
	The denominator is clearly positive, so we focus on the sign of the numerator, which is positive when
	\begin{equation*}
		\begin{split}
			\frac{1 - w_1}{w_1} \, \frac{c}{1 - c} \, \frac{\sigma_1 + \rho \sigma_2}{\rho \sigma_2} \, \phi\bigg(\frac{-c\mu}{\rho \sigma_2}\bigg) & < \phi\bigg(\frac{-(1-c)\mu}{\sigma_1+\rho \sigma_2}\bigg)
			\\ \iff \frac{1 - w_1}{w_1} \, \frac{c}{1 - c} \, \frac{\sigma_1 + \rho \sigma_2}{\rho \sigma_2} & < \exp\bigg(-\bigg(\bigg(\frac{-(1-c)\mu}{\sigma_1+\rho \sigma_2}\bigg)^2 - \bigg(\frac{-c\mu}{\rho \sigma_2}\bigg)^2 \, \bigg) \bigg/ 2 \bigg).
		\end{split}
	\end{equation*}
	We first check the sign of the exponent, which is negative when
	$$\frac{-(1-c)\mu}{\sigma_1+\rho \sigma_2} > \frac{-c\mu}{\rho \sigma_2}  \iff c > \frac{\rho \sigma_2}{\sigma_1 + 2\rho \sigma_2} = c^*.$$ Note that for $w_1 < 1/2$, discovering policy 1 is always worse for the principal than if $w_1 = 1/2$, regardless of other parameters\textemdash reducing $w_1$ makes spoiling relatively more important. Thus in this case the cutoff $\underline{c}$ is bounded above $(\rho \sigma_2)/(\sigma_1 + 2 \rho \sigma_2)$, so the exponent is negative at all relevant values of $c$. Furthermore, because $c/(1 - c)$ is increasing in $c$, we can use the lower bound $\underline{c} \geq c^*$ to ensure the left-hand side is bounded above
	$$\frac{1 - w_1}{w_1} \, \frac{c^*}{1 - c^*} \, \frac{\sigma_1 + \rho \sigma_2}{\rho \sigma_2} = \frac{1 - w_1}{w_1}.$$
	Because $w_1 < 1/2$, the left-hand side is bounded strictly above 1. Thus because the right-hand side is bounded weakly below 1, the condition fails for all $\mu$ and the cutoff is decreasing.
	
	If instead $w_1 > 1/2$, then by similar logic to the previous case, the exponent will be positive, so the right-hand side is bounded above 1; the left-hand side can be upper-bounded below $1$ by using the upper bound $\underline{c} = c^*$ and simplifying; and we conclude that the cutoff is increasing.
	
	We now finish our characterization. If $(1 - w_1)/(2w_1) < 1/2 \iff w_1 > 1/2$, then case (2) occurs for small enough $\mu$, whereas case (1) never occurs because $w_1/(2-2w_1) > 1/2$. Thus for large $w_1$, there is $\bar{\mu} \in \mathbb{R}_+$ and an increasing cutoff $\underline{c}(\mu)$ such that $\underline{c}(\bar{\mu}) = 0$ and $\lim_{\mu \rightarrow \infty} \underline{c}(\mu) \leq (\rho \sigma_2)/(\sigma_1 + 2 \rho \sigma_2)$. Discovering project 1 is best if and only if $c > \underline{c}(\mu)$.
	
	If instead $w_1 < 1/2$ then case (1) attains for small $\mu$ and case (2) never attains, even in the limit. Then there is \underbar{$\mu$}$\in \mathbb{R}_+$ and a decreasing cutoff $c^*(\mu)$ such that $\underline{c}(\underline{\mu}) = 1$ and $\lim_{\mu \rightarrow \infty} \underline{c}(\mu) \geq (\rho \sigma_2)/(\sigma_1 + 2 \rho \sigma_2)$. Discovering project 1 is best if and only if $c > \underline{c}(\mu)$.
	
	Finally, we know that in fact both limits must be precisely equal to $c^*$. In the limit, because the normal cdf has exponential tails, so the ratio
	$$\Phi\bigg(\frac{-c\mu}{\rho \sigma_2}\bigg) / \Phi\bigg(\frac{-(1-c)\mu}{\sigma_1+\rho \sigma_2}\bigg)$$ 
	is either 0 if $c > c^*$ and the numerator dominates, $\infty$ if $c < c^*$ and the denominator dominates, or 1 if $c = c^*$. Thus in order for this ratio to match the value $(1 - w_1/w_1)$ as $\mu$ grows large, the cutoff $\underline{c}$ must converge to $c^*$ (and in fact must do so at an exponential rate). Because increasing $w_1$ (regardless of whether it is below or above $1/2$) for fixed $\mu$ weakly decreases $\bar{c}$, it must be that the rate of convergence of $\underline{c}(\mu)$ to $c^*$ is increasing in $|w_1 - 1/2|$.
\end{proof}
	
\begin{proof}[Proof of Theorem \ref{thm:optSingle}] $\text{} $
	
	For this theorem, we begin by proving two auxiliary results. The first is a corollary of Theorem \ref{thm:trialBalloon} which applies that approach to compare discovering project 2 to no discovery:
	\begin{corollary}
		\label{cor:comparison}
		Let $\mu_1 = -\mu < 0 \leq c \mu = \mu_2$ for $c \in [0, 1)$.\\
		For any $\mu$, there is a cutoff $\underline{c}(\mu)$ such that it is better to discover project 2 than to discover neither project if and only if $c > \underline{c}(\mu)$.\\
		The function $\underline{c}(\mu)$ is continuous and has the following properties:
		\begin{enumerate}
			\item If $w_1 = 1/2$, then
			$$\underline{c}(\mu) = c^{**} := \frac{\sigma_2}{\rho \sigma_1 + 2\sigma_2} \qquad \text{for all } \mu.$$
			\item If $w_1 > 1/2$, there is $\mu^* \in \mathbb{R}_+$ such that $\underline{c}(\mu) = 0$ for all $\mu \in [0, \mu^*]$. For $\mu > \mu^*$, the cutoff $\underline{c}(\mu)$ is strictly increasing with
			$\lim_{\mu \rightarrow \infty} \underline{c}(\mu) = c^{**}.$
			\item If $w_1 < 1/2$, there is $\mu^* \in \mathbb{R}_+$ such that $\underline{c}(\mu) = 1$ for all $\mu \in [0, \mu^*]$. For $\mu > \mu^*$, the cutoff $\underline{c}(\mu)$ is strictly decreasing with
			$\lim_{\mu \rightarrow \infty} \underline{c}(\mu) = c^{**}.$
		\end{enumerate}
	\end{corollary}
	\begin{proof}
		To compare discovery of project 2 to no discovery, we emulate the proof of Theorem \ref{thm:trialBalloon}. We omit the algebraic details, which are exactly as in that proof, and focus on the key conclusions.
		
		The cutoff in the knife-edge case $w_1 = 1/2$ is given by
		$$\frac{-(1-c)\mu}{\rho \sigma_1+\sigma_2} > \frac{-c\mu}{\sigma_2}  \iff c > \frac{\sigma_2}{\rho \sigma_1 + 2 \sigma_2} := c^{**}.$$
		
		Once again decreasing $w_1$ harms the principal (relative to no discovery) and increasing $w_1$ helps her, so the cutoff for $w_1 < 1/2$ is increasing in $\mu$ and converges to $c^{**}$, whereas the cutoff for $w_2$ is decreasing in $\mu$ and converges to $c^{**}$.
	\end{proof}

	The second result is a characterization of whether discovering project 1 or discovering project 2 dominates in the limit:
	\begin{lemma}
		\label{lm:smallVar}	
		Let $\mu_1 = -\mu < 0 \leq c \mu = \mu_2$, for $c \in [0, 1)$ and let $\sigma_1 < \sigma_2$, so that project 1 is disfavored but not a trial balloon. Then,
		\begin{enumerate}
			\item For any $c \in (c^*, c^{**})$ there exists $\underline{\mu}$ such that discovering project 1 is better than discovering project 2 for all
			$\mu > \underline{\mu}$.
			\item For any $c \in (c^{**}, 1)$ there exists $\bar{\mu}$ such that discovering project 2 is better than discovering project 1 for all $\mu > \bar{\mu}$.
		\end{enumerate}
	\end{lemma} 
	\begin{proof}
		We make use of the expressions $\Pi^1$ and $\Pi^2$ from the proof of Theorem \ref{thm:trialBalloon}. 
		
		When $c \in (c^*, c^{**})$ we know from Theorem \ref{thm:trialBalloon} that there exists $\mu'$ such that discovering project 1 is preferred to no discovery for all $\mu > \mu'$, and from Corollary \ref{cor:comparison} that there exists $\mu''$ such that no discovery is preferred to discovering project 2 for all $\mu > \mu''$. Thus there exists $\bar{\mu} = \max \left\{\mu', \mu''\right\}$ such that discovering project 1 is preferred to discovering project 2. For $c < c^*$, we know no discovery will dominate so there is no need to establish a limit result.
		
		When $c > c^{**}$, we know that discovering either project is eventually preferred to no discovery. However, because of the exponential nature of the normal distribution, we also know that the improvement probability from discovering project 2 dominates\footnote{By dominance we mean that the likelihood ratio between the improvement probability and the other specified probability becomes arbitrarily large.} not only the spoiling probability from discovering either project, but also the improvement probability from discovering project 1. Thus for $\mu$ large enough, discovering project 2 is better than discovering project 1.
	\end{proof}

	We now combine these two intermediate results to state the conclusion of the theorem. To make precise the cutoffs described in the text, we divide the result into the cases $w_1 = 1/2$, $w_1 < 1/2$, and $w_1 > 1/2$. We also describe a fourth region of very low $\mu$ for the latter two cases.
	
	First, when $w_1 = 1/2$, the cutoffs are constant at $c^*$ and $c^{**}$. Thus the low-$\mu$ case described in the main text vanishes. In the high-$\mu$ case, there is no discovery if $c < c^*$, discovery of project 1 if $c \in (c^*, c^{**})$, and discovery of an unknown project if $c > c^{**}$. In the limit case, there is no discovery for $c < c^*$, discovery of project 1 for $c \in (c^*, c^{**})$, and discovery of project 2 for $c > c^{**}$.
	
	When $w_1 < 1/2$,
	\begin{enumerate}
		\item Very low $\mu$: both cutoffs are 1, so there is no discovery for all $c$.
		\item Low $\mu$: This region is lower-bounded by the infimum value of $\mu$ at which $\min \left\{c_1(\mu), c_2(\mu)\right\} < 1$. There is discovery for high enough $c$ but it is unknown which project is discovered at each $c$.
		\item High $\mu$: This region is lower-bounded by the value of $\mu$ where $c_1(\mu) = c^{**}$. Because $c_1(\mu)$ is decreasing, there is no discovery for $c < c_1(\mu)$, discovery of project 1 for $c \in (c_1(\mu), c_2(\mu))$, and discovery of an unknown project for $c > c_2(\mu)$.
		\item Limit $\mu$: there is no discovery for $c < c^*$, discovery of project 1 for $c \in (c^*, c^{**})$, and discovery of project 2 for $c > c^{**}$.
	\end{enumerate}
	
	Finally, when $w_1 > 1/2$,
	\begin{enumerate}
		\item Very low $\mu$: both cutoffs are 0, so there is discovery of an unknown project for all $c$.
		\item Low $\mu$: This region is lower-bounded by the infimum value of $\mu$ at which $\max \left\{c_1(\mu), c_2(\mu)\right\} > 0$. There is discovery for high enough $c$ but it is unknown which project is discovered at each $c$.
		\item High $\mu$: This region is lower-bounded by the value of $\mu$ where $c_2(\mu) = c^{*}$. Because $c_2(\mu)$ is increasing, there is no discovery for $c < c_1(\mu)$, discovery of project 1 for $c \in (c_1(\mu), c_2(\mu))$, and discovery of an unknown project for $c > c_2(\mu)$.
		\item Limit $\mu$: there is no discovery for $c < c^*$, discovery of project 1 for $c \in (c^*, c^{**})$, and discovery of project 2 for $c > c^{**}$.
	\end{enumerate}

	In fact, in the high-$\mu$ cases, we can ensure that for some neighborhood of $c_2(\mu)$, discovery of project 1 is still optimal. This conclusion follows because when $c = c_2(\mu)$, the principal is precisely indifferent between discovering project 2 and no discovery. However, the principal strictly prefers discovering project 1 to no discovery. Because all utilities are continuous, it must be that for some $\varepsilon$-neighborhood above $c = c_2(\mu)$, discovering project 1 remains optimal.
\end{proof}

We now present a new result which establishes that the dominance region for project 1 contracts in the large-$\mu$ limit if $w_1 < 1/2$ and expands if $w_1 > 1/2$.
\begin{lemma}[Limit Behavior of the Dominance Region for Project 1]
	\label{lm:compStat}
	Let $\mu_1 = -\mu < 0 \leq c \mu = \mu_2$ for $c \in [0, 1)$ and let $\sigma_1 < \sigma_2$. Then the following are true:
	\begin{itemize}
		\item There is $\underline{\mu}$ large enough such that, for all $\mu > \underline{\mu}$, the dominance region where project 1 is discovered differs by at most measure $\varepsilon > 0$ from $(c_1(\mu), c_2(\mu))$.
		\item For all $\mu > \underline{\mu}$, if $w_1 < 1/2$ then the measure of that dominance region is strictly decreasing in $\mu$.
		\item For all $\mu > \underline{\mu}$, if $w_1 > 1/2$ then the measure of that dominance region is strictly increasing in $\mu$.
	\end{itemize}
\end{lemma}
\begin{proof}[Proof of Lemma \ref{lm:compStat}]  $\text{} $
	Thanks to Lemma \ref{lm:smallVar} we know that once $\mu > \underline{\mu}$, for any $c > c_2(\mu) \geq c^{**}$, project 2's greater improvement probability is the most important force driving disclosure; that is, it is first-order whereas all other effects are second-order. Thus we have two conclusions. First, the mass of $c > c_2(\mu)$ in the dominance region of project 1 is vanishing as $\mu$ gets large, because only the improvement effect of project 2 will dominate the payoff of no discovery. This result is the first bullet point in the lemma. Second, it cannot be that some $c > c_2(\mu)$ is in the dominance region of project 1 for for $\mu > \underline{\mu}$ but not for $\mu' > \mu$, because the benefits of discovering project 1 will vanish faster than those of discovering project 2; the vanishing established in the first bullet point is therefore monotonic. We thus restrict attention to the behavior of the interval $(c_1(\mu), c_2(\mu))$.
	
	To verify that the range of $c$ where project 1 is discovered narrows with increasing $\mu$ when $w_1 < 1/2$, note that there are two effects: values $c < c_1(\mu)$ are added as $c_1(\mu)$ decreases and values $c > c_2(\mu)$ are removed as $c_2(\mu)$ decreases. In the limit, the rate of decrease of each is driven by the vanishing spoiling effect, so we can write
	\footnotesize
	$$c_1'(\mu) \approx \frac{d}{d\mu} \Phi\bigg(\frac{-c\mu}{\rho \sigma_2}\bigg) = -\frac{c}{\rho \sigma_2} \, \phi\bigg(\frac{-c\mu}{\rho \sigma_2}\bigg) \qquad \text{and} \qquad c_2'(\mu) \approx \frac{d}{d\mu} \Phi\bigg(\frac{-c\mu}{\sigma_2}\bigg) = -\frac{c}{\sigma_2} \, \phi\bigg(\frac{-c\mu}{\sigma_2}\bigg).$$\normalsize
	In the limit the difference in the $\phi(\cdot)$ terms overcomes the difference in coefficients, and so $c_1(\mu)$ decreases slower than $c_2(\mu)$. Thus more values are removed from the range $(c_1(\mu), c_2(\mu))$ than added.
	
	When $w_1 > 1/2$, values $c < c_2(\mu))$ are added as $c_2(\mu)$ increases and values $c < c_1(\mu)$ are removed as $c_1(\mu)$ increases. Both increases are driven by vanishing improvement, so we have
	\begin{equation*}
		\begin{split}
			& c_1'(\mu) \approx \frac{d}{d\mu} \Phi\bigg(\frac{-(1 - c)\mu}{\sigma_1 + \rho \sigma_2}\bigg) = -\frac{(1 - c)}{\sigma_1 + \rho \sigma_2} \, \phi\bigg(\frac{-(1 - c)\mu}{\sigma_1 +  \rho \sigma_2}\bigg) \qquad \text{and} 
			\\& c_2'(\mu) \approx \frac{d}{d\mu} \Phi\bigg(\frac{-(1 - c)\mu}{\rho \sigma_1 + \sigma_2}\bigg) = -\frac{(1 - c)}{\rho \sigma_1 + \sigma_2} \, \phi\bigg(\frac{-(1 - c)\mu}{\rho \sigma_1 + \sigma_2}\bigg),
		\end{split}
	\end{equation*}
	and $c_1(\mu)$ increases slower in the limit so the overall range where project 1 is discovered increases, similar to the 
	
	These properties may not hold away from the limit because we can no longer approximate the discovery region of project 1 by the interval $(c_1(\mu), c_2(\mu))$, and even if we could the rate of change of $c_1(\mu)$ and $c_2(\mu)$ depend on both spoiling and improvement rather than only one of the two effects.
\end{proof}

\begin{proof}[Proof of Theorem \ref{thm:both}] $\text{} $
		
	If $\mu_1 < 0$ and $\mu_2 < 0$, we may decompose utility as in Proposition \ref{prop:disfavor} to get
	\begin{equation*}
		\begin{split}
			\mathbb{P}(\text{both approved}) & + w_1 \, \mathbb{P}(\text{only 1 approved}) + w_2 \, \mathbb{P}(\text{only 2 approved})
			\\& = \frac{1}{2} \, \mathbb{P}(\text{both}) + \frac{1}{2} \, \mathbb{P}(\text{at least 1}).
		\end{split}
	\end{equation*}
	Then, as in the proof of that result, we may combine Propositions \ref{prop:bundle} and \ref{prop:subs} to get the desired conclusion.
		
	Now let $\mu_1 = -\mu < 0 \leq c \mu = \mu_2$ for $c \in [0, 1)$. Assuming $\sigma_1 \geq \sigma_2$ allows us to apply Theorem \ref{thm:trialBalloon} and ignore the possibility of discovering project 2. We break this proof into two lemmas. The first establishes dominance regions in the interval $[0, 1)$ of possible values of $c$ for discovering both projects or discovering only project 1. The second compares discovering both projects to no discover.
	
	\begin{lemma}
		\label{lm:bothVs1}
		Let $\mu_1 = -\mu < 0 \leq c \mu = \mu_2$ for $c \in [0, 1)$. Then, if 
		$$c < c_\ell := \frac{\rho \sigma_2}{\rho \sigma_2 + \sqrt{\sigma_1^2 + \sigma_2^2 +2\rho \sigma_1 \sigma_2}},$$
		there exists $\bar{\mu}$ such that discovering project 1 is better than full discovery for all $\mu > \bar{\mu}$. If instead
		$$c > c_h := \frac{\sigma_2}{\sigma_1 + (1 + \rho) \, \sigma_2},$$
		then there exists $\underline{\mu}$ such that full discovery is better than discovering project 1 for all $\mu > \underline{\mu}$.
	\end{lemma}
	\begin{proof}
		Let $X$ be a $N(0, 1)$ random variable. Using the decomposition above, the payoff for discovering both projects is
		\begin{equation*}
			\begin{split}
				\Pi^b & = \frac{1}{2} \, \mathbb{P}\bigg(X < \frac{(1-c)\mu}{\sqrt{\sigma_1^2 + \sigma_2^2 +2 \rho \sigma_1 \sigma_2}}\bigg) - \frac{1}{2} \big(1 - \mathbb{P}(v_1 < 0 \text{ and } v_2 < 0)\big)
				\\ & \leq 1 - \frac{1}{2} \, \Phi\bigg(\frac{(1-c)\mu}{\sqrt{\sigma_1^2 + \sigma_2^2 +2 \rho \sigma_1 \sigma_2}}\bigg) - \frac{1}{2} \, \Phi\bigg(\frac{-c\mu}{\sigma_2}\bigg) \, \Phi\bigg(\frac{\mu}{\sigma_1}\bigg),
			\end{split}
		\end{equation*}
		The inequality comes from the fact that because the $\rho > 0$, the probability that  both $v_1$ and $v_2$ are less than some value is bounded above by the product of the probabilities.
		
		Because we cannot write the spoiling probability using a univariate normal cdf, we cannot follow the the approach of Lemma \ref{lm:smallVar}. Instead, we directly compare the expression $\Pi^1$ from Theorem \ref{thm:trialBalloon} to our upper bound on $\Pi^b$:
		\begin{equation*}
			\footnotesize
			\begin{split}
				1 - \frac{1}{2} \, \Phi\bigg(\frac{(1-c)\mu}{\sigma_1+\rho \sigma_2}\bigg)- \frac{1}{2} \, \Phi\bigg(\frac{-c\mu}{\rho \sigma_2}\bigg) > & \, 1 - \frac{1}{2} \, \Phi\bigg(\frac{(1-c)\mu}{\sqrt{\sigma_1^2 + \sigma_2^2 +2\rho \sigma_1 \sigma_2}}\bigg) 
				\\& - \frac{1}{2} \, \Phi\bigg(\frac{\mu}{\sigma_1}\bigg) \, \Phi\bigg(\frac{-c\mu}{\sigma_2}\bigg)
				\\ \iff \Phi\bigg(\frac{(1-c)\mu}{\sqrt{\sigma_1^2 + \sigma_2^2 +2\rho \sigma_1 \sigma_2}}\bigg) + \Phi\bigg(\frac{\mu}{\sigma_1}\bigg)\Phi\bigg(\frac{-c\mu}{\sigma_2}\bigg) > & \, \Phi\bigg(\frac{(1-c)\mu}{\sigma_1+\rho \sigma_2}\bigg) +  \Phi\bigg(\frac{-c\mu}{\rho \sigma_2}\bigg)
				\\ \iff \Phi\bigg(\frac{-c\mu}{\sigma_2}\bigg) - \Phi\bigg(\frac{-c\mu}{\rho \sigma_2}\bigg)  - \bigg(1-\Phi\bigg(\frac{\mu}{\sigma_1}\bigg)\bigg) \, \Phi\bigg(\frac{-c\mu}{\sigma_2}\bigg) > & \, \Phi\bigg(\frac{(1-c)\mu}{\sigma_1+\rho \sigma_2}\bigg) 
				\\ & - \Phi\bigg(\frac{(1-c)\mu}{\sqrt{\sigma_1^2 + \sigma_2^2 +2\rho \sigma_1 \sigma_2}}\bigg)
				\\ \iff \int_{-c\mu/(\rho \sigma_2)}^{-c\mu/\sigma_2} \, \phi(x) \, dx - \bigg(1-\Phi\bigg(\frac{\mu}{\sigma_1}\bigg)\bigg) \Phi\bigg(\frac{-c\mu}{\sigma_2}\bigg) > & \, \int_{(1 - c) \mu / \sqrt{\sigma_1^2 + \sigma_2^2 +2\rho \sigma_1 \sigma_2}}^{(1 - c) \mu / (\sigma_1 + \rho \sigma_2)} \, \phi(x) \, dx.
			\end{split}
		\end{equation*}
		The term $(1-\Phi(\mu/\sigma_1)) \, \Phi(-c\mu/\sigma_2)$ vanishes with order $(\mu \exp(\mu^2))^2$. However, the integrals vanish only with order $\mu \exp(\mu^2)$. Thus for $\mu$ sufficiently large we can directly compare the integrals. We also use the symmetry of the normal distribution to switch from the integral over $(-c\mu/(\rho \sigma_2), -c\mu/\sigma_2)$ to the integral over $(c\mu/\sigma_2, c\mu/(\rho \sigma_2))$:
		\begin{equation*}
			\tiny
			\begin{split}
				\int_{c\mu/\sigma_2}^{c\mu/(\rho \sigma_2)} \, \phi(x) \, dx > & \int_{(1 - c) \mu / \sqrt{\sigma_1^2 + \sigma_2^2 +2\rho \sigma_1 \sigma_2}}^{(1 - c) \mu / (\sigma_1 + \rho \sigma_2)} \, \phi(x) \, dx
				\\ \impliedby  \phi\bigg(\frac{c\mu}{\rho \sigma_2}\bigg) \, c \, \mu \, \frac{1-\rho}{\rho \sigma_2} > & \, \phi\bigg(\frac{(1-c)\mu}{\sqrt{\sigma_1^2 + \sigma_2^2 +2\rho \sigma_1 \sigma_2}}\bigg) \, (1-c) \, \mu \, \bigg(\frac{1}{\sigma_1+\rho \sigma_2}- \frac{1}{\sqrt{\sigma_1^2 + \sigma_2^2 +2\rho \sigma_1 \sigma_2}}\bigg)
				\\ \iff \phi\bigg(\frac{c\mu}{\rho \sigma_2}\bigg) \bigg/ > & \, \frac{1 - c}{c} \, \frac{\rho \sigma_2}{1 - \rho} \, \bigg(\frac{1}{\sigma_1+\rho \sigma_2}- \frac{1}{\sqrt{\sigma_1^2 + \sigma_2^2 +2\rho \sigma_1 \sigma_2}}\bigg) \phi\bigg(\frac{(1-c)\mu}{\sqrt{\sigma_1^2 + \sigma_2^2 +2\rho \sigma_1 \sigma_2}}\bigg) := k
				\\ \iff k < & \exp\bigg(-\bigg(\bigg(\frac{c\mu}{\rho \sigma_2}\bigg)^2 - \bigg(\frac{(1-c)\mu}{\sqrt{\sigma_1^2 + \sigma_2^2 +2\rho \sigma_1 \sigma_2}}\bigg)^2 \,\bigg)\bigg/2 \bigg)	
				\\ \impliedby \text{for large } \mu \qquad \bigg(\frac{c\mu}{\rho \sigma_2}\bigg)^2 < & \, \bigg(\frac{(1-c)\mu}{\sqrt{\sigma_1^2 + \sigma_2^2 +2\rho \sigma_1 \sigma_2}}\bigg)^2
				\\ \iff c < & \, \frac{\rho \sigma_2}{\rho \sigma_2 + \sqrt{\sigma_1^2 + \sigma_2^2 +2\rho \sigma_1 \sigma_2}} =: c_\ell.
			\end{split}
		\end{equation*}
		We can now re-introduce the term $(1-\Phi(\mu/\sigma_1)) \, \Phi(-c\mu/\sigma_2)$. Our limit result states that for $\mu$ large enough, the inequality in the first line holds. Beyond this point, the difference between the integrals vanishes with the same order as each integral on its own; that is, with order $\mu \exp(\mu^2)$. Thus the term $(1-\Phi(\mu/\sigma_1)) \, \Phi(-c\mu/\sigma_2)$ will be exponentially smaller than the difference between the integrals and does not alter the inequality.
		
		Now we perform the reverse bounding exercise to show that discovering both projects is better than discovering project 1 for high $c$. In place of $\Pi^b$ we use a lower bound on the utility from discovering both, obtained by subtracting $\mathbb{P}(v_2 < 0)$ rather than $\mathbb{P}(v_1 < 0 \text{ and } v_2 < 0)$. Using $\mathbb{P}(v_1 < 0)$ would result in a looser bound.
		\begin{equation*}
			\footnotesize
			\begin{split}
				1 - \frac{1}{2} \, \Phi\bigg(\frac{(1-c)\mu}{\sigma_1+\rho \sigma_2}\bigg)- \frac{1}{2} \, \Phi\bigg(\frac{-c\mu}{\rho \sigma_2}\bigg) < & \, 1 - \frac{1}{2} \, \Phi\bigg(\frac{(1-c)\mu}{\sqrt{\sigma_1^2 + \sigma_2^2 +2\rho \sigma_1 \sigma_2}}\bigg) - \frac{1}{2} \, \Phi\bigg(\frac{-c\mu}{\sigma_2}\bigg)
				\\ \iff \Phi\bigg(\frac{(1-c)\mu}{\sqrt{\sigma_1^2 + \sigma_2^2 +2\rho \sigma_1 \sigma_2}}\bigg) + \Phi\bigg(\frac{-c\mu}{\sigma_2}\bigg) < & \, \Phi\bigg(\frac{(1-c)\mu}{\sigma_1+\rho \sigma_2}\bigg) +  \Phi\bigg(\frac{-c\mu}{\rho \sigma_2}\bigg)
				\\ \iff \Phi\bigg(\frac{-c\mu}{\sigma_2}\bigg) - \Phi\bigg(\frac{-c\mu}{\rho \sigma_2}\bigg) < & \, \Phi\bigg(\frac{(1-c)\mu}{\sigma_1+\rho \sigma_2}\bigg) - \Phi\bigg(\frac{(1-c)\mu}{\sqrt{\sigma_1^2 + \sigma_2^2 +2\rho \sigma_1 \sigma_2}}\bigg)
				\\ \iff \int_{-c\mu/(\rho \sigma_2)}^{-c\mu/\sigma_2} \, \phi(x) \, dx < & \, \int_{(1 - c) \mu / \sqrt{\sigma_1^2 + \sigma_2^2 +2\rho \sigma_1 \sigma_2}}^{(1 - c) \mu / (\sigma_1 + \rho \sigma_2)} \, \phi(x) \, dx
				\\ \iff \int_{c\mu/\sigma_2}^{c\mu/(\rho \sigma_2)} \, \phi(x) \, dx < & \, \int_{(1 - c) \mu / \sqrt{\sigma_1^2 + \sigma_2^2 +2\rho \sigma_1 \sigma_2}}^{(1 - c) \mu / (\sigma_1 + \rho \sigma_2)} \, \phi(x) \, dx
				\\ \impliedby \frac{1}{k} < & \, \phi\bigg(\frac{(1 - c) \mu}{\sigma_1 + \rho \sigma_2}\bigg) \bigg/ \phi\bigg(\frac{c \mu}{\sigma_2}\bigg)
				\\ \iff \frac{1}{k} < & \, \exp\bigg(-\bigg( \bigg(\frac{(1-c)\mu}{\sigma_1 + \rho \sigma_2}\bigg)^2 - \bigg(\frac{c\mu}{\sigma_2}\bigg)^2 \bigg) \bigg/ 2\bigg)
				\\ \impliedby \text{for large } \mu \qquad \bigg(\frac{(1-c)\mu}{\sigma_1 + \rho \sigma_2}\bigg)^2 < & \, \bigg(\frac{c\mu}{\sigma_2}\bigg)^2
				\\ \iff c_h := \frac{\sigma_2}{\sigma_1 + (1 + \rho) \, \sigma_2} < & \, c.
			\end{split}
		\end{equation*}
		Note the similarity of this expression to $c^*$ and $c^{**}$ in Theorem \ref{thm:optSingle}. The difference comes from comparing the improvement probability associated with discovering project 1 to the spoiling probability associated with discovering project 2. Thus this bound lies above $c^{*}$ (because spoiling is stronger than when discovering project 1) but below $c^{**}$ (because improvement is weaker than when discovering project 2). 
	\end{proof}
	
	The second lemma compares no discovery to discovering both projects:
	\begin{lemma}
		\label{lm:bothVsNo}
		Let $\mu_1 = -\mu < 0 \leq c \mu = \mu_2$ for $c \in [0, 1)$. There is a cutoff $c_\text{no}$ such that
		\begin{itemize}
			\item For any $c \in [0, c_\text{no})$, there exists $\bar{\mu}(c)$ such that full discovery is better than no discovery if and only if $\mu \in [0, \bar{\mu}(c))$. 
			\item For any $c > c_\text{no}$, full discovery is better than no discovery for all $\mu$.
		\end{itemize}
	\end{lemma}
	\begin{proof}
		Studying the contour sets of the normal distribution shows that for any $c < 1/2$, discovering neither is eventually better than discovering both, and once that occurs, that dominance is monotonic. For $\rho \in (0, 1)$, the cutoff obtained for discovering project 1 to dominate discovering both projects is below $1/2$ so this argument always applies.
		
		Formally, we can plot the contour sets of the standard normal in $\mu_1$-$\mu_2$ space as ellipses slanted along the line $y = x$ and centered at the origin. Discovering both projects can be seen as choosing a contour set, then drawing uniformly from the contour set. On each contour set, we can draw three regions:
		\begin{enumerate}
			\item Above the line $y = -x$, where the sum of project values is positive.
			\item In the third quadrant, where both project values are negative.
			\item The remaining region (two symmetric triangles in the first and fourth quadrants) where exactly one project value is positive.
		\end{enumerate}
		For the standard normal, we can see that this decomposition yields a positive expected value for each contour set.
		
		The contour sets of a normal with different means can be derived by shifting the center appropriately. As we shift the center to the left along the $x$-axis, meaning that project 1's mean becomes more negative, regions (2) and (3) expand to cover region (1) without increasing the relative share of region (3) compared to region (2). Thus expected payoff decreases. If we instead shift the center up along the $y$-axis, meaning that project 2's mean becomes more positive, regions (1) and (3) expand to cover region (2) without increasing the relative share of region (3) compared to region (1). Thus expected payoff increases. Therefore, as we move along some curve $c_\text{no} = c(\mu)$ in the second quadrant, regions (1) and (2) become and remain equal in size for all ellipses not containing the origin. For such contour sets, the expected payoff is $1/2$; however, because the distribution has full support, some contour sets will cover the origin. The payoff in those contour sets is strictly larger, because in each of them region (1) is larger than region (2). Thus the payoff from full discovery for any project with mean lying on the curve $c_\text{no}$ exceeds $1/2$, the payoff from no discovery. For $c > c_\text{no}$, the payoff from revealing both increases, so full discovery continues to dominate. For $c < c_\text{no}$, the region (2) is larger than the region (1) for ellipses not containing the origin. For $\mu$ large, mass is concentrated in these ellipses; because the normal distribution has vanishing tails, we can eventually ignore the mass from ellipses containing the origin (i.e., for any $\varepsilon > 0$ there is $\mu$ large enough that the gains from full discovery captured by the origin-containing contour sets are below $\varepsilon$). By taking $\mu$ large enough that those $\varepsilon$-gains are bounded below the losses from the contour sets not containing the origin, we ensure that full discovery is worse than no discovery. Therefore for any $c < c_\text{no}$, there is $\mu$ large enough that no discovery is better than no discovery. The payoffs from full discovery are monotonically declining in $\mu$, so in fact for $c < c_\text{no}$ there is a cutoff $\bar{mu}(c)$ such that full discovery is preferred if and only if $\mu \in [0, \bar{\mu})$.
	\end{proof}
	
	We then combine the two lemmas to state the theorem. Note that we do not bound $c_\text{no}$ in relation to $c_\ell$ and $c_h$, or indeed provide a closed-form solution, leading to the presentation in the text.
\end{proof}

\begin{proof}[Proof of Proposition \ref{prop:noise}] $\text{} $
	
	We begin by computing the distribution of the agent's posterior mean belief about $v_i$, which we call $q_i$, after observing the signal $s_i$. Because $s_i$ is normally distributed, as is the prior over $v_i$, we can use standard properties of the normal distribution to get
	$$q_i \sim N\bigg(\mu_i, \frac{\sigma_i^4}{\tau_i^2 + \sigma_i^2}\bigg).$$
	Intuitively, as the signal $s_i$ becomes noisier, the agent pays less attention to it, so beliefs do not stray far from the prior mean $\mu_i$. In the limiting case $\tau_i^2 = 0$, the distribution is the same as if the agent observed public discovery of $v_i$ (as in Section \ref{sec:bvNorm}). In the other limiting case $\tau_i^2 \rightarrow \infty$, the signal is completely uninformative and the agent always places full weight on the prior mean $\mu_i$.
	
	Although this distribution describes the possible posterior mean beliefs after observing $s_i$, note that the agent's actual posterior about $v_i$ will have distribution
	$$v_i \sim N(q_i, \sigma_i^2 + \tau_i^2).$$
	Thus we get the distribution
	$$v_j | s_i \sim N\bigg(\mu_j + \rho \frac{\sigma_j}{\sqrt{\sigma_i^2 + \tau_i^2}} \, (s_i - \mu_i), (1 - \rho)^2 \sigma_j^2\bigg)$$
	for the agent's posterior belief about $v_j$.
	
	Thus we can replicate all the results of Section \ref{sec:bvNorm} using the distribution 
	$$N\bigg(\mu_i, \frac{\sigma_i^4}{\tau_i^2 + \sigma_i^2}\bigg)$$
	for the project $i$ which the principal chooses to discover and 
	$$N\bigg(\mu_j + \rho \frac{\sigma_j}{\sqrt{\sigma_i^2 + \tau_i^2}} \, (s_i - \mu_i), (1 - \rho)^2 \sigma_j^2\bigg)$$
	for the posterior belief about the project $j$ conditional on discovery of project $i$. A noisy signal weakens both the spoiling and improvement effects for project $i$ by reducing the amount by which discovery shifts the agent's beliefs.
\end{proof}

\begin{proof}[Proof of Proposition \ref{prop:Nballoon}] $\text{} $
	
	We begin by characterizing the principal's utility in a general $N$-project setting, with no restrictions on project means or variances. Extending the computation in the proof of Proposition \ref{prop:bundle}, we see that the probability of obtaining a bundle containing projects $j \in S$ for some subset of projects $S$, upon revealing project $i$, is
	\begin{equation*}
		\begin{split}
			\mathbb{P}\bigg(v_i + \sum_{j \in S} \mathbb{E}[v_j | v_i] \geq 0\bigg) & = \mathbb{P}\bigg(v_i \geq \frac{\rho \mu_i \sum_{j \in S} (\sigma_j/\sigma_i) - \sum_{j \in S} \mu_j}{1 + \rho \sum_{j \in S} (\sigma_j/\sigma_i)}\bigg)
			\\& = \Phi\bigg(\frac{\mu_i + \sum_{j \in S} \mu_j}{\sigma_i + \rho \sum_{j \in S} \sigma_j}\bigg)
		\end{split}
	\end{equation*}
	Similarly, extending the decomposition of utility from the proof of Theorem \ref{thm:trialBalloon}, we get that the principal's utility is given by
	$$\sum_{i \in N} w_i \, \mathbb{P}(\text{some bundle containing } i).$$
	Define $S_j^i$ as
	$$S_j^i := \bigg(\arg \max_{S \subseteq N} \mathbb{P}(\text{bundle containing } j \text{ approved} \mid i \text{ discovered})\bigg) \setminus \left\{i, j\right\},$$
	that is, the most likely bundle containing $j$ conditional on project $i$ being discovered, excluding project $i$ and $j$ themselves. To further parallel Theorem \ref{thm:trialBalloon}, let $\mu_i = -\mu < 0$ and let $\mu_j = c_j \mu$ for constants $c_j \in \mathbb{R}$. Then we can write the principal's utility from discovering project $i$ as
	\footnotesize
	$$w_i \, \Phi\bigg(\frac{-(1 - \sum_{j \in S_i^i} c_j) \, \mu}{\sigma_i + \rho \sum_{j \in S_j^i} \sigma_j}\bigg) + \sum_{j \text{ s.t. } c_j < 0} w_j \, \Phi\bigg(\frac{-(1 - c_j - \sum_{j \in S_j^i} c_j) \, \mu}{\sigma_i + \rho \sigma_j + \rho \sum_{j \in S_j^i} \sigma_j}\bigg) - \sum_{j \text { s.t. } c_j \geq 0} w_j \, \Phi\bigg(\frac{- c_j \mu}{\rho \sigma_j}\bigg).$$\normalsize
	An approach in the style of Corollary \ref{cor:comparison} would allow us to compute the utility from discovering a project $j$ with $c_j \geq 0$ as well.
	
	From this expression, we can see that the approach to comparing discovery of project $i$ with no discovery is largely preserved. Fixing all parameters except some $c_j$, the first two terms are weakly increasing in $c_j$ whereas the third is weakly decreasing in $c_j$. Thus there is once again a well-defined continuous cutoff function $\underline{c}(\mu) \in \mathbb{R}$ such that the sum of the first two terms exceeds the third if and only if $c_j > \underline{c}_j(\mu)$. Unfortunately, we cannot analytically sign $\underline{c}_j'(\mu)$ to determine whether the range of viable $c_j$ increases or decreases with higher precision; we expect that the result will be highly dependent on the remaining parameters of the model.
	
	Now we impose the assumptions stated in the text to simplify the expression above:
	$$w_i \, \Phi\bigg(\frac{-(1 - \sum_{j \in S_i^i} c_j) \, \mu}{\sigma_i + \rho \sum_{j \in S_i^i} \sigma_j}\bigg) + \sum_{j \neq i} w_j \, \Phi\bigg(\frac{-c_j \mu}{\rho \sigma_j}\bigg).$$
	If another project $k$ is discovered, we have utility
	$$w_i \, \Phi\bigg(\frac{-(1 - c_k - \sum_{j \in S_i^k} c_j) \, \mu}{\rho \sigma_i + \sigma_k + \rho \sum_{j \in S_i^k} \sigma_j}\bigg) + \Phi\bigg(\frac{-c_k \mu}{\sigma_k}\bigg) + \sum_{j \notin \left\{i, k\right\}} w_j \, \Phi\bigg(\frac{-c_j \mu}{\rho \sigma_j}\bigg).$$
	Discovering project $i$ is preferred whenever
	\footnotesize
	$$w_i \, \Phi\bigg(\frac{-(1 - \sum_{j \in S_i^i} c_j) \, \mu}{\sigma_i + \rho \sum_{j \in S_i^i} \sigma_j}\bigg) +  w_k \, \Phi\bigg(\frac{-c_k \mu}{\rho \sigma_j}\bigg) > w_i \, \Phi\bigg(\frac{-(1 - c_k - \sum_{j \in S_i^k} c_j) \, \mu}{\rho \sigma_i + \sigma_k + \rho \sum_{j \in S_i^k} \sigma_j}\bigg) + \Phi\bigg(\frac{-c_k \mu}{\sigma_k}\bigg).$$\normalsize
	Clearly the second term on the left-hand side is larger than the second term on the right-hand side, as in the proof of Theorem \ref{thm:trialBalloon}. If project $i$ also satisfies $\sigma_i = \max_{j \in N} \sigma_j$, then we can show that the first term on the left-hand side is greater than the first term on the right-hand side, so discovering project $i$ is always better than discovering any other project:
	$$\frac{-(1 - \sum_{j \in S_i^i} c_j) \, \mu}{\sigma_i + \rho \sum_{j \in S_i^i} \sigma_j} \geq \frac{-(1 - c_k - \sum_{j \in S_i^k} c_j) \, \mu}{\sigma_i + \rho \sigma_k + \rho \sum_{j \in S_i^k} \sigma_j} \geq \frac{-(1 - c_k - \sum_{j \in S_i^k} c_j) \, \mu}{\rho \sigma_i + \sigma_k + \rho \sum_{j \in S_i^k} \sigma_j}.$$
	The first inequality follows because $S_i^i$ is by construction a better choice of bundle when discovering project $i$ than $S_i^k \cup \left\{k\right\}$. The second follows from $\sigma_i \geq \sigma_k$.
	Thus discovering a trial balloon remains optimal in this setting.
	
	Finally, we examine the content of Theorem \ref{thm:optSingle} in this setting. If project $i$ remains the only negative-mean project but $\sigma_i < \max_{j \in N} \sigma_j$, then we cannot in general sign the value
	$$\delta = \bigg(\frac{-(1 - \sum_{j \in S_i^i} c_j) \, \mu}{\sigma_i + \rho \sum_{j \in S_i^i} \sigma_j}\bigg) - \bigg(\frac{-(1 - c_k - \sum_{j \in S_i^k} c_j) \, \mu}{\rho \sigma_i + \sigma_k + \rho \sum_{j \in S_i^k} \sigma_j}\bigg).$$
	If $\delta > 0$ for all parameter values despite the assumption on $\sigma_i$, then Proposition \ref{prop:Nballoon} still holds. Otherwise, in the cases where $\delta < 0$ we have the same trade-off as in Theorem \ref{thm:optSingle}, where project $i$ has a lower spoiling chance than any project $k$, but project $k$ has a greater improvement chance. We can find limit cutoffs in the style of $c^*$ and $c^{**}$ in that theorem for each $c_j$ with $j \in S_i^i \cup S_i^k \cup \left\{k\right\}$, and compare them if $c_j \in (S_i^i \cup \left\{k\right\}) \cap (S_i^k \cup \left\{k\right\})$ to obtain a version of Theorem \ref{thm:optSingle} in this more general setting. However, it is harder to visualize these results in the style of Figure \ref{fig:optSingle} because the space $\left\{c_j\right\} \times \mu$ is now $N$-dimensional.
	
	If project $i$ is no longer the only negative-mean project, we now have to deal with the potential for $S_j^i \neq S_j^k$ for all negative-mean projects $j$. That is, we now have an expression $\delta_j$ in the style of $\delta$ above comparing improvement effects across discovery rules for each negative-mean project. In general, we cannot sign these expressions without explicitly fixing parameters, and cannot even make a comparison of limit cutoffs using the techniques of Theorem \ref{thm:optSingle} because there are multiple improvement effects to consider. 
\end{proof}

\end{document}